\documentclass[fleqn,11pt]{article}
\usepackage{amscd,amsmath,amssymb,verbatim}
\usepackage{amsthm}
\usepackage[pdftex]{graphicx}
\usepackage{float}

\usepackage{authblk}

\usepackage{multirow}

\usepackage[T1]{fontenc}
\usepackage{array}
\usepackage{mathrsfs}
\usepackage[nohead]{geometry}
\usepackage[singlespacing]{setspace}
\usepackage{rotating}
\usepackage{pdflscape}
\usepackage[table,xcdraw]{xcolor}

\usepackage[authoryear]{natbib}

\usepackage{hyperref}

\newtheorem{theorem}{Theorem}

\newtheorem{assumption}{Assumption}

\newtheorem{remark}{Remark}
\newtheorem{lemmaB}{Lemma}[section]

\newtheorem{theoremB}{Theorem}[section]

\newtheorem{corollaryB}{Corollary}[section]

\hoffset 
\voffset

\date{\today}
\title{A Heterogeneous Spatiotemporal GARCH Model: A Predictive Framework for Volatility in Financial Networks}

\author[1]{Atika Aouri}
\author[2]{Philipp Otto\footnote{Corresponding author: philipp.otto@glasgow.ac.uk}}

\affil[1]{Abdelhafid Boussouf University Center, Mila, Algeria}
\affil[2]{University of Glasgow, School of Mathematics and Statistics, UK}

\begin{document}

\maketitle

\abstract{We introduce a heterogeneous spatiotemporal GARCH model for geostatistical data or processes on networks, e.g., for modelling and predicting financial return volatility across firms in a latent spatial framework. The model combines classical GARCH($p$, $q$) dynamics with spatially correlated innovations and spatially varying parameters, estimated using local likelihood methods. Spatial dependence is introduced through a geostatistical covariance structure on the innovation process, capturing contemporaneous cross-sectional correlation. This dependence propagates into the volatility dynamics via the recursive GARCH structure, allowing the model to reflect spatial spillovers and contagion effects in a parsimonious and interpretable way. In addition, this modelling framework allows for spatial volatility predictions at unobserved locations. In an empirical application, we demonstrate how the model can be applied to financial stock networks. Unlike other spatial GARCH models, our framework does not rely on a fixed adjacency matrix; instead, spatial proximity is defined in a proxy space constructed from balance sheet characteristics. Using daily log returns of 50 publicly listed firms over a one-year period, we evaluate the model’s predictive performance in a cross-validation study.}

\section{Introduction}

Volatility clustering is a well-known phenomenon in financial time series: large return innovations tend to be followed by periods of high volatility, while tranquil periods often coincide with reduced market activity. Generalised autoregressive conditional heteroskedasticity (GARCH) models \citep{engle1982autoregressive, bollerslev1986generalized} provide a robust framework for capturing such temporal dependence in volatility. However, in real-world financial systems, individual assets are not isolated. Instead, they are embedded in complex financial networks shaped by common exposures, interbank lending, or structural economic similarities \citep[see, e.g.,][]{diebold2023past, chen2019tail, mattera2024network}. Consequently, volatility patterns often exhibit pronounced dependence across firms, sectors, or regions. This phenomenon is typically referred to interchangeably as spatial, cross-sectional, network, or topological dependence, depending on the context and the underlying representation of relationships. Similar forms of volatility propagation arise in environmental and epidemiological contexts, where uncertainty in model predictions, such as pollutant concentrations or disease incidence rates, varies dynamically across both space and time due to complex interactions between spatial covariates, temporal dynamics, and measurement errors \citep{huang2011class,otto2018generalised,otto2023dynamic,rodeschini2024scenario}.

This motivates the development of spatiotemporal GARCH models, which incorporate spatial interaction in the volatility dynamics. Existing approaches can be grouped into four broad classes: (1) multivariate GARCH models that include spatial lags in the volatility equation but assume cross-sectional independence (in volatility) at the same time point \citep[e.g.,][]{borovkova2012spatial,caporin2015proximity,holleland2020stationary,billio2023networks}, (2) spatiotemporal GARCH or log-GARCH models with instantaneous spatial dependence in volatility \citep[e.g.,][]{sato2017spatial,sato2018spatial,sato2018spatiotemporal,otto2023general,otto2024dynamic}, (3) stochastic volatility models with spatially correlated innovations \citep{huang2011class}, and (4) the novel idea proposed in this paper where spatial dependence enters through the conditional mean equation, indirectly propagating volatility through autoregressive feedback. 

Most of the spatiotemporal GARCH models have been developed for areal data or fixed sets of spatial locations (a notable exception is the stochastic volatility model of \cite{huang2011class}), limiting their ability to predict volatility at unobserved sites. We refer the interested reader to \cite{otto2025spatial} for a review of spatial and spatiotemporal volatility models. In contrast, geostatistical models operate on continuous spatial domains and enable interpolation (kriging) of latent variables at arbitrary locations. While kriging is widely used in mean modelling, its extension to volatility models remains underexplored. Bridging this gap requires a framework that not only accounts for spatial volatility dynamics but also enables volatility prediction at unobserved locations, a task with clear relevance in financial risk monitoring, stress testing, and systemic risk assessment, but also in other fields such as environmental science (e.g., modelling and forecasting pollution variability across monitoring networks) or epidemiology, where it supports the reliable estimation of spatial prediction intervals of disease outbreak intensities or uncertainty in infection dynamics.

Another limitation of existing models is their assumption of spatially homogeneous GARCH parameters, which may be unrealistic in heterogeneous domains. In financial networks, for instance, firms with different balance sheet structures may exhibit distinct levels of baseline volatility, persistence, and responsiveness to shocks. Similar forms of spatial heterogeneity arise in environmental applications (e.g., urban vs. rural air quality dynamics) and epidemiology (e.g., region-specific disease transmission patterns). Allowing GARCH parameters to vary smoothly across space reflects these differences and provides a more flexible and interpretable representation of volatility dynamics. Such parameter surfaces have been previously explored in nonparametric time series models \citep{Dahlhaus}, and we extend this idea to a spatiotemporal GARCH framework. Most existing spatiotemporal GARCH models rely on the assumption of spatial stationarity, which may be overly restrictive in the presence of heterogeneity, asymmetry, or directional effects \citep{Atika}. To address this, we build on the probabilistic results of \citet{Atika}, who extend classical GARCH dynamics by allowing spatially varying parameters and introducing geostatistical dependence via a covariance structure on the innovation process. This captures contemporaneous cross-sectional correlations that propagate forward in time through the recursive GARCH mechanism, enabling a parsimonious yet flexible representation of spatial volatility spillovers.

In this paper, we propose a heterogeneous spatiotemporal GARCH model defined over a continuous spatial domain. Volatility is modelled using a standard GARCH($p$,$q$) recursion at each spatial location, with spatially varying parameters and geostatistical dependence in the innovation process. Estimation is based on a locally weighted quasi-likelihood approach that yields smooth parameter surfaces and enables kriging of the volatility at unobserved locations. We illustrate the proposed methodology in an empirical application to a panel of 50 global firms with daily stock returns. More precisely, we consider the spatial domain to be represented by firm-level financial indicators, such as balance sheet positions, which induce a latent economic space over which the GARCH process is defined. This setting allows us to model and predict volatility across financial networks, also for unobserved locations/firms. We demonstrate how the resulting model captures spatial heterogeneity in volatility, propagates risk through the network, and enables accurate prediction of volatility at unobserved firms. Our findings highlight the importance of both spatial dependence and heterogeneity in modelling firm-level risk dynamics.

The remainder of the paper is organised as follows. Section \ref{sec:theory} introduces the heterogeneous spatiotemporal GARCH model. Sections \ref{sec:kriging} and \ref{sec:estimation} present the local quasi-likelihood estimation procedure and discusses kriging for volatility prediction. In Section \ref{sec:MC}, we evaluate the model's performance through a Monte-Carlo simulation study. Section \ref{sec:application} provides the empirical analysis of financial network data. Section \ref{sec:conclusion} concludes with a discussion and outlook.

\section{Heterogeneous Spatiotemporal GARCH model}\label{sec:theory}

We consider a spatiotemporal process $\{Z_t(s): s \in \mathcal{D},\, t \in \mathcal{T}\}$, defined over a general spatial domain $\mathcal{D} \subseteq \mathcal{M}$, where $\mathcal{M}$ denotes an underlying spatial structure such as a Euclidean space, a manifold (e.g., a sphere), or a network. This formulation allows for broad applicability: for instance, $\mathbb{R}^2$ is suitable for most environmental or epidemiological data with geographic coordinates, the sphere $\mathbb{S}^2$ accommodates global climate models, and financial or transportation networks are naturally represented by graph-based domains, where distances reflect structural dissimilarity rather than physical proximity. To ensure validity across such domains, the spatial covariance function governing the innovations must respect the geometry of $\mathcal{M}$ \citep[see, e.g.,][]{gneiting2010matern, porcu2016spatio, porcu2023stationary}, and the local estimation kernels must be adapted accordingly. In the present work, without loss of generality, we focus on a Euclidean domain $\mathcal{D} = [0,1]^2 \subset \mathbb{R}^2$ and assume a discrete, equidistant time index $t \in  \mathcal{T} = \mathbb{Z}$. We define a spatiotemporal GARCH($p$,$q$) (STGARCH) process as follows:
\begin{equation}\label{eq1}
\begin{cases}
Z_t(s) = \sigma_t(s) \eta_t(s)  \\ 
\sigma_t^2(s) = \omega(s) + \sum_{i=1}^q \alpha_i(s) Z_{t-i}^2(s) + \sum_{j=1}^p \beta_j(s) \sigma_{t-j}^2(s) \\
\end{cases}
\end{equation}
where $\omega(s)$, $\alpha_i(s)$, and $\beta_j(s)$ are location-specific parameters that vary smoothly over space, with the orders $q < p$. The innovation process $\{\eta_t(s)\}$ is assumed to be independent across time, strictly stationary and ergodic, with zero mean and unit variance, and exhibits spatial dependence through a valid covariance function on $\mathcal{D}$. In what follows, we consider a covariance function of the Mat\'{e}rn class \citep{gneiting2010matern}. This geostatistical structure captures contemporaneous cross-sectional correlations in the innovation process, which subsequently propagate into future volatility via the recursive GARCH equation. Thus, the model allows for spatial heterogeneity in both the innovation structure and the volatility dynamics, while supporting prediction at unobserved locations. Under mild smoothness conditions on the parameter functions, the spatially nonstationary process \(Z_t(s)\) can be locally approximated at any fixed location $s_0\in[0,1]^2$ by a spatially stationary GARCH($p$,$q$) process 
with constant parameters evaluated at $s_0$.  In particular, define
\begin{equation*}
\begin{cases}
\widetilde{Z}_{t}(s)=\widetilde{\sigma}_{t}(s)\,\eta_{t}(s), \\
\widetilde{\sigma}_{t}^2(s)
= \omega(s_{0})\;+\;\sum_{i=1}^q\alpha_i(s_{0})\widetilde{Z}_{t-i}^2(s)\;+\;\sum_{j=1}^p\beta_j(s_{0})\widetilde{\sigma}_{t-j}^2(s). \\
\end{cases}
\end{equation*}
This locally stationary approximation underpins our asymptotic analysis in Section \ref{sec:properties}. Further details on the approximation can be found in Appendix \ref{B}.

\begin{remark}\label{rmk:1}
Although the volatility process $\sigma_t^2(s)$ in \eqref{eq1} does not include instantaneous spatial interactions (i.e., there are no terms of the form $\sigma_t^2(s')$ with $s' \neq s$), and is, thus, temporally causal at each location $s$, it nonetheless inherits spatial dependence through the squared observations $Z_{t-i}^2(s)$. Since $Z_t(s) = \sigma_t(s)\eta_t(s)$ and the innovation process $\eta_t(s)$ is spatially correlated across $s \in \mathcal{D}$, the process $Z_t(s)^2$ also exhibits spatial dependence. Importantly, the range of spatial dependence in $Z_t(s)^2$ differs from that in $\eta_t(s)$. For instance, if $\eta_t(s)$ follows a zero-mean, unit-variance Gaussian process with spatial covariance function $C(h)$, it holds that $\mathrm{Cov}(\eta_t(s)^2, \eta_t(s')^2) = 2 \cdot C(h)^2,$
where $h = \|s - s'\|$. Thus, the spatial correlation in the squared process $\eta_t(s)^2$ decays faster than in $\eta_t(s)$, since $C(h)^2$ drops off more quickly with increasing $h$ than $C(h)$. This effect carries over to the squared process $Z_t(s)^2$, and consequently to $\sigma_t^2(s)$ through the recursive GARCH dynamics. Therefore, even though $\sigma_t^2(s)$ is not explicitly defined via spatial interactions at the same time point, it is still a spatially dependent process, with the dependence structure shaped by the properties of the innovation field.
\end{remark}

\subsection{Model Assumptions}\label{assumptions}

To develop the asymptotic theory for our spatiotemporal GARCH process defined in Equation~\eqref{eq1}, we impose a set of assumptions that ensure temporal stationarity, control the spatial variation in parameters, and guarantee the existence of necessary moments. We begin by formalising the properties of the innovation process.

\begin{assumption}\label{ass1}
  The innovations $ \eta_t(s) $ are independent over time, and spatially, strictly stationary. 
\end{assumption}

This assumption reflects the temporal white noise nature of the innovations, while allowing for spatial dependence, captured through a valid geostatistical covariance function. Stationarity in space ensures that the  structure of the innovation process is spatially invariant, an important requirement for spatial prediction and asymptotic consistency. Next, we define the parameter space and regularity conditions that ensure the existence of a strictly stationary solution to the spatiotemporal GARCH recursion.

\begin{assumption}\label{ass2} 
Let \( \Omega \) be the compact set defined by:
\begin{equation}
    \Omega= \{ (\omega(s), \alpha_i(s), \beta_j(s)) : \sum_{i=1}^{q} \alpha_i(s) + \sum_{j=1}^{p} \beta_j(s) \leq 1 - \gamma, \rho_1 \leq \omega(s) \leq \rho_2 \text{ for } i = 1, \ldots, q \text{ and } j = 1, \ldots, p \},
\end{equation}
where \(0 < \rho_1 \leq \rho_2 < \infty\) and $\gamma$ is a positive constant. For each \(s \in (0,1]^2\), we assume that \(\Theta\) is in the interior of \(\Omega\), where \(\Theta = (\omega(s), \alpha_1(s), \ldots, \alpha_p(s), \beta_1(s),\dots ,\beta_q(s))\).
\end{assumption}

This assumption guarantees that the local GARCH dynamics are well-defined, with the stationarity condition that $\sum \alpha_i(s) + \sum \beta_j(s)$ is bounded away from the unit circle,  uniformly across space. It is worth noting that this condition also implies that $ \inf_{s \in \mathcal{D}} \sigma_t^2(s) > \rho_1$, i.e., it is uniformly bounded away from zero, to avoid degeneracy of the conditional variance. Moreover, the parameter surfaces need to be smooth across space. This smoothness is crucial for the local estimation procedure and for ensuring uniform convergence of estimators. This assumption allows us to approximate the process locally by a stationary GARCH model and facilitates the application of kernel-based smoothing methods in estimation.

\begin{assumption}\label{ass4} 
The second-order partial derivatives of $ \omega_t(s)$, $\alpha_i(s)$ and $\beta_j(s)$ exist and all their partial derivatives up to the second order are uniformly bounded over $[0,1]^2.$ 
\end{assumption}

Typically for stationary GARCH processes, the result can be proved if the fourth moments of the innovations are finite. 
We require a mildly stronger moment condition on the innovation process that is slightly stronger than the standard GARCH setting. This is necessary to apply central limit results for martingale arrays, which arise due to the locally weighted quasi-likelihood estimation \citep[cf.][, Theorem 3.2]{hall2014martingale}.

\begin{assumption}\label{ass5} 
For some $v>0$, $E\left(\left|\eta _i(s)\right|^{4+v}\right)<\infty$. 
\end{assumption}

The assumption \ref{ass2}, \ref{ass4} and \ref{ass5}  are needed to prove asymptotic normality of the estimators defined below. In geostatistical applications, a common choice for the innovation process $\eta_t(s)$ is a Gaussian process with mean zero and a Matérn covariance function, which satisfies all required conditions. In this case, all moments exist and Assumption~\ref{ass5} is trivially fulfilled.

\subsection{Prediction at Unobserved Locations}\label{sec:kriging}

Before turning to estimation, we first outline the kriging problem for spatially nonstationary spatiotemporal GARCH models as defined in equation~\eqref{eq1}. Our goal is to predict the process $Z_t(s_0)$ at an arbitrary, unobserved location $s_0 \in \mathcal{D}$, based on the available observations $\{Z_t(s_u): t = 1, \dots, T; u = 1, \dots, m\}$. This prediction, known as spatial interpolation or kriging in geostatistics, requires knowledge of the location-specific parameters $\omega(s_0), \alpha_i(s_0), \beta_j(s_0)$ that define the conditional variance $\sigma_t^2(s_0)$. While these parameter surfaces are not known in practice, it is conceptually helpful to first develop the prediction framework under the assumption that they are available. This approach mirrors classical kriging setups, where the spatial covariance structure is assumed known in the derivation of the best linear unbiased predictor. In Section~\ref{sec:estimation}, we return to the problem of estimating the parameters $\omega(\cdot), \alpha_i(\cdot), \beta_j(\cdot)$ via a locally weighted likelihood approach. For now, we focus on the construction of the kriging predictor and its associated uncertainty.

To facilitate prediction, we first derive an ARMA-type representation of the STGARCH model, which enables straightforward computation of linear predictors. We assume that the innovations $\eta_t(s)$ follow a spatially stationary Gaussian process, consistent with the common geostatistical modelling approach. For each fixed $s$, the process $\{Z_t(s)\}$ then constitutes a stationary GARCH($p$,$q$) time series under the conditions in Assumption~\ref{ass2}. More precisely, this implies that the squared process $Z_t^2(s)$ satisfies the following stationary ARMA representation, i.e.,
\begin{equation}\label{eq:arma}
    Z_t^2(s) = \omega(s) + \sum_{i=1}^{r} \delta_i(s) Z_{t-i}^2(s) + \sum_{j=1}^{p} \beta_j(s) \zeta_{t-j}(s) + \zeta_t(s),
\end{equation}
where $r = \max(p,q)$, and $\delta_i(s) = \alpha_i(s) + \beta_i(s)$, with $\alpha_i(s) = 0$ for $i > q$ and $\beta_i(s) = 0$ for $i > p$. The innovation term is defined as $\zeta_t(s) = \sigma_t^2(s)(\eta_t^2(s) - 1)$, which is uncorrelated over time and has zero mean under the assumption $\mathbb{E}[\eta_t^2(s)] = 1$ (follows straightforward from the unit variance of $\eta_t(s)$).

Hence, we can write $Z_t^2(s)$ in the following form 
\begin{equation}
    Z_t^2(s) =  \sum_{j=0}^{\infty} \phi_j(s) \zeta_{t-j}(s) +  C(s)   
\end{equation}
and the sequence of non-parametric functions  $\phi_j(s)$ is given by 
\begin{equation}
    \phi_j(s) =  \beta_j(s) + \sum_{i=1}^j \delta_i(s) \phi_{j-i}(s),\;\;  j \geq 0 
\end{equation}
where $ \beta_0(s)=  1 $ , $ \beta_j(s) =0 $ for $ j >p $ , $\phi_j(s) = 0 $ if $ j < 0 $ . In addition 
\begin{equation}
    C(s) = \omega(s) \Big( 1- \sum_{j=1}^r \delta_j(s) \Big)^{-1}.
\end{equation}
Thus, $Z_t^2(s)$ can be rewritten as follows:
\begin{equation} \label{eq2}
Z_t^2(s) = \sum_{k=0}^{\infty} \nu_k(s) \zeta_{t-k}(s)+ C(s), 
\end{equation}
where $\nu_k(s) =  [A(s)^k]_{1,1}$, with
\begin{equation}
    A(s) = \begin{pmatrix}
\phi_1(s) &\phi_2(s) & \cdots & \phi_{p-1}(s)& \phi_p(s)\\
1 & 0 & \cdots & 0 & 0 \\
0& 1&\ddots & 0& 0 \\
\vdots& \vdots&  \ddots&\vdots& \vdots \\
0 &\cdots& \cdots & 1& 0\end{pmatrix}
\end{equation}
and we assume that 
\begin{equation}
    \sup_{s} |\nu_k(s)| \leq K \lambda^k.
\end{equation}

Let  $\mathbf{Z}=\left[\left\{Z_j^2(s_u): u=1, \ldots, m\right\}_{j=1}^T\right]$ denote the observed squared process across all spatial locations and time points. Based on the ARMA representation in equation~\eqref{eq2}, we evaluate the conditional expectation $\mathbb{E}\left[Z_t^2\left(s_0\right) \mid \mathbf{Z}\right]$ for $p+1 \leq t \leq T$.  To this end, we consider the vector-valued process
\begin{equation}
\underline{Z}_t^2\left(s_0\right)=A\left(s_0\right) \underline{Z}_{t-1}^2\left(s_0\right)+\underline{\zeta}_t\left(s_0\right) + \underline{C}(s_0),
\end{equation}
where
\begin{eqnarray}
\underline{Z}_t^2\left(s_0\right)^{\prime}=\left(Z_t^2\left(s_0\right), \ldots, Z_{t-p+1}^2\left(s_0\right)\right), \; \zeta_t\left(s_0\right)^{\prime}=\left(\zeta_t\left(s_0\right) , \ldots, 0\right), \; \text{and} \nonumber \\
\underline{C}(s_0) =\left(C(s_0) , \ldots, 0\right) . 
\end{eqnarray}

To derive the kriging predictor for the latent volatility process at an unobserved location $s_0$, we consider the conditional expectation $\mathbb{E}[Z_t^2(s_0) \mid \mathbf{Z}]$, which provides the minimum mean-square error prediction of $Z_t^2(s_0)$ given the observed spatiotemporal process. Starting from the ARMA representation in equation~\eqref{eq:arma}, we can express this conditional expectation as
\begin{equation}
   \mathbb{E}\left[Z_t^2\left(s_0\right) \mid \mathbf{Z}\right]=\sum_{k=0}^{t-p-1} \nu_k\left(s_0\right) \mathbb{E}\left[\zeta_{t-k}\left(s_0\right) \mid \mathbf{Z}\right]+\left[A\left(s_0\right)^{t-p} \mathbb{E}\left\{\underline{Z}_p^2\left(s_0\right) \mid \mathbf{Z}\right\}\right]_1 + C(s_0) 
\end{equation}
Here, $\mathbb{E}[\zeta_{t-k}(s_0) \mid \mathbf{Z}]$ is evaluated by conditioning on $\{\zeta_{t-k}(s_u): u = 1, \dots, m\}$ using the spatial covariance structure of the innovation process, which enables interpolation of volatility innovations across space. Specifically, for $p+1 \leq t-k \leq T$, we express
\begin{equation}
\mathbb{E}\left[\zeta_{t-k}\left(s_0\right) \mid \mathbf{Z}\right]=\mathbb{E}\left[\zeta_{t-k}\left(s_0\right) \mid\left\{\zeta_{t-k}\left(s_u\right): u=1, \ldots, m\right\}\right]
\end{equation}
and substitute it into the expression above to obtain 
\begin{eqnarray}
\mathbb{E}\left[Z_t^2\left(s_0\right) \mid \mathbf{Z}\right]= & \sum_{k=0}^{t-p-1} \nu_k\left(s_0\right) \mathbb{E}\left[\zeta_{t-k}\left(s_0\right) \mid\left\{\zeta_{t-k}\left(s_u\right): s=1, \ldots, m\right\}\right] \\
& +\left[A\left(s_0\right)^{t-p} \mathbb{E}\left\{\underline{Z}_p^2\left(s_0\right) \mid \mathbf{Z}\right\}\right]_1     + C(s_0)  .
\end{eqnarray}
Under Assumption  \ref{ass2}, we have the bound $\left\|A\left(s_0\right)^k\right\|_{\text {spec }} \leq K  \lambda^k$, where $\| \cdot\|_{\text{spec}}$ denotes the spectral norm, which means
\begin{equation}\label{eq3}
\mathbb{E}\left[Z_t^2\left(s_0\right) \mid \mathbf{Z}\right] = \sum_{k=0}^{t-p-1} \nu_k\left(s_0\right) \mathbb{E}\left[\zeta_{t-k}\left(s_0\right) \mid\left\{\zeta_{t-k}\left(s_u\right): u=1, \ldots, m\right\}\right]+  C(s_0)   + O_p\left(\lambda^{t-p}\right)
\end{equation}
Further details are provided in Appendix \ref{B}.

Next, we demonstrate that the first term on the right-hand side of \eqref{eq3} can be evaluated recursively. Suppose that the random process $\left\{Z_t^2\left(s_0 \mid m\right)\right\}$ satisfies
\begin{eqnarray}\label{eq4}
Z_t^2\left(s_0 \mid m\right) & = & \omega(s_0) + \sum_{i=1}^r \delta_i(s_0) Z_{t-i}^2(s_0 \mid m) \\
&& + \; \sum_{j=1}^p \beta_j(s_0) \mathbb{E}\left[\zeta_{t-j} \left(s_0\right) \mid \zeta_t\left(s_1\right), \ldots, \zeta_t\left(s_m\right)\right] \nonumber   \\
&& + \; \mathbb{E}\left[\zeta_t\left(s_0\right) \mid \zeta_t\left(s_1\right), \ldots, \zeta_t\left(s_m\right)\right] \nonumber  
\end{eqnarray}
where $Z_t^2\left(s_0 \mid m\right)=0$ for $t \leq p$. We define the $p$-dimensional vectors
$$
\begin{aligned}
\underline{Z}_t^2\left(s_0 \mid m\right)^{\prime} & =\left(Z_t^2\left(s_0 \mid m\right), \ldots, Z_{t-p+1}^2\left(s_0 \mid m\right)\right) \\
\underline{\zeta}_t\left(s_0 \mid m\right)^{\prime} & =\left( \mathbb{E}\left[\zeta_t\left(s_0\right) \mid \zeta_t\left(s_1\right), \ldots, \zeta_t\left(s_m\right)\right], 0, \ldots, 0\right) .
\end{aligned}
$$

We note  that \eqref{eq4} can be written as $\underline{Z}_t^2\left(s_0 \mid m\right)=A\left(s_0\right) \underline{Z}_{t-1}^2\left(s_0 \mid m\right)+\underline{\zeta}_t\left(s_0 \mid m\right) + C(s_0) $. Given the initial condition $\underline{Z}_p^2\left(s_0 \mid m\right)^{\prime}=(0, \ldots, 0)$, the solution for $Z_t^2\left(s_0 \mid m\right)$ is
\begin{equation}
    \sum_{k=0}^{t-p-1} \nu_k\left(s_0\right) \mathbb{E}\left[\zeta_{t-k}\left(s_0\right) \mid\left\{\zeta_{t-k}\left(s_u\right): u=1, \ldots, m\right\}\right]+  C(s_0).
\end{equation}
Comparing $Z_t^2\left(s_0 \mid m\right)$ and \eqref{eq3}, we observe that $Z_t^2\left(s_0 \mid m\right)$ and $\mathbb{E}\left[Z_t^2\left(s_0\right) \mid \mathbf{Z}\right]$ are asymptotically equivalent as $t$ becomes large. As a result, $Z_t^2\left(s_0 \mid m\right)$ can be utilised as an estimator for $\mathbb{E}\left[Z_t^2\left(s_0\right) \mid \mathbf{Z}\right]$.

Evaluating  $Z_t^2\left(s_0 \mid m\right)$ requires determining   $\mathbb{E}\left[\zeta_t\left(s_0\right) \mid \zeta_t\left(s_1\right), \ldots, \zeta_t\left(s_m\right)\right]$, along with the non-parametric functions  $\omega(\cdot), \alpha_j(\cdot), \text{and} \; \beta_i(\cdot)$. Furthermore, the conditional expectation \linebreak $\mathbb{E}\left[\zeta_t\left(s_0\right) \mid\left\{\zeta_t\left(s_u\right)\right\}_{u=1}^m\right]$ can only be computed in closed form if the distribution of \linebreak $\left\{\zeta_t\left(s_u\right): u=0, \ldots, m\right\}$ is known. As the Gaussian process is the most widely adopted distribution in geostatistics, and facilitates closed-form derivations, we consider this special case in more details below. In the case where the innovation process $\{\eta_t(s) \; ;  \;  s \in[0,1]^2 \}$ follow a spatial Gaussian process, then $\zeta_t(s)$ is a spatial centred $\chi^2$-distributed process. We can demonstrate that 
$$
\mathbb{E}\left[\zeta_t\left(s_0\right) \mid\left\{\zeta_t\left(s_u\right)\right\}_{u=1}^m\right]=\sum_{u=1}^m \gamma_j(\underline{s}) \zeta_t\left(s_u\right),
$$
$$  \mathbb{E}\left[Z_t^2\left(s_0\right) \mid \mathbf{Z}\right]   =\sum_{k=0}^{t-p-1} \nu_k\left(s_0\right) \sum_{u=1}^m \gamma_u(\underline{s}) \zeta_{(t-k)}\left(s_u\right) + C(s_0).$$
where
$$
\underline{s}=\left(s_0, \ldots, s_m\right), \quad \gamma(\underline{s})=\left(\gamma_1(\underline{s}), \ldots, \gamma_m(\underline{s})\right)^{\prime}=R(\underline{s})^{-1} r(\underline{s})
$$
with
$$
R(\underline{s})_{i j}=\operatorname{cov}\left(\zeta_t\left(s_i\right), \zeta_t\left(s_j\right)\right)=c_{\zeta}\left(s_i-s_j\right) \text { and } r(\underline{s})_i=\operatorname{cov}\left(\zeta_t\left(s_0\right), \zeta_t\left(s_i\right)\right)=c_{\zeta}\left(s_i-s_0\right) \text {.}
$$
This implies that by utilising estimates of the innovations $\left\{\zeta_t\left(s_u\right)\right\}_t$ at the observed locations, we can estimate the covariance function $c_{\zeta}\left(\cdot\right)$, e.g., using a maximum-likelihood approach. Recall that, as noted in Remark~\ref{rmk:1}, if $\eta_t(s)$ follows a Gaussian process with covariance $C_{\eta}(h)$, then $\eta_t^2(s)-1$ has covariance $2\,C_{\eta}(h)^2$, implying a shorter effective spatial range. Consequently, $\zeta_t(s)$ inherits this reduced range whenever $\sigma_t^2(s)$ varies smoothly over space. In practice, we recommend to back-transform the estimated $\zeta_t(s)$ to $\eta_t(s)$, i.e., $\eta_t(s) = \text{sign}(Z_t(s))\sqrt{\frac{\zeta_t(s)}{\sigma_t^2(s)} + 1}$.

\begin{remark}\label{rmk:2}
In practice, the quantities $\omega(\cdot)$, $\alpha_j(\cdot)$, and $\beta_i(\cdot)$ required in \eqref{eq4} are unknown and must be estimated from the data. One possibility is to estimate all parameter surfaces jointly with the spatial covariance parameters of $\zeta_t(s)$ in a single optimisation step, which fully accounts for the spatio–temporal dependence but can be computationally demanding in high–dimensional settings.  
Alternatively, we suggest adopting a computationally simpler two–stage strategy: we first estimate the local spatiotemporal GARCH parameter surfaces independently at each location, then compute the residuals $\{\eta_t(s)\}$, and finally estimate the spatial covariance parameters from these residuals. This approach is common in geostatistical regression and retains good statistical properties while being substantially faster than the joint estimation alternative.
\end{remark}

Finally, we derive the covariance of the kriging predictor for the squared spatiotemporal GARCH process \( Z_t^2(s_0) \) at an unobserved spatial location $s_0$. Based on the recursive representation and spatial kriging approximation of the innovations, the predictor is given by 
\begin{equation}
    \text{Cov}(\hat{Z}_t^2(s_0), \hat{Z}_{t'}^2(s_0)) 
= \text{Cov}\left( \sum_{k=0}^{t-p-1} \nu_k(s_0) \sum_{u=1}^m \gamma_u(s_0)\, \zeta_{t-k}(s_u),\,
                    \sum_{\ell=0}^{t'-p-1} \nu_\ell(s_0) \sum_{v=1}^m \gamma_v(s_0)\, \zeta_{t'-\ell}(s_v) \right). \notag 
\end{equation}
Since \( C(s_0) \) is nonrandom, it does not contribute to the covariance. Due to the temporal independence of \( \zeta_t(\cdot) \),  and if $ t-k = t'-\ell $ we get       
\begin{equation}
\text{Cov}(\hat{Z}_t^2(s_0), \hat{Z}_{t'}^2(s_0))  = \sum_{k=0}^{t-p-1} \sum_{\ell=0}^{t'-p-1} \nu_k(s_0)\nu_\ell(s_0)
    \sum_{u=1}^m \sum_{v=1}^m \gamma_u(s_0)\gamma_v(s_0)\,
    \text{Cov}(\zeta_{t-k}(s_u), \zeta_{t'-\ell}(s_v)).
\end{equation}
which simplifies to 
\begin{equation}
\text{Cov}(\hat{Z}_t^2(s_0), \hat{Z}_{t'}^2(s_0)) = \sum_{k=0}^{t-p-1} \sum_{\ell=0}^{t'-p-1} \nu_k(s_0)\nu_\ell(s_0) \sum_{u=1}^m \sum_{v=1}^m \gamma_u(s_0)\gamma_v(s_0)\, R_{uv}.
\end{equation}
The resulting expression yields the closed-form covariance of the kriging predictor for $Z_t^2(s_0)$, determined by the temporal ARMA coefficients $\nu_k(s_0)$, the spatial kriging weights $\gamma_u(s_0)$, and the spatial covariance matrix $R$ of the innovation process. Since these quantities depend on unknown parameter surfaces $\omega(\cdot)$, $\delta_i(\cdot)$, $\beta_j(\cdot)$ as well as on the spatial covariance of $\zeta_t(s)$, their practical implementation requires consistent estimation from the observed data. The following section outlines the estimation strategy for these components, which provides the foundation for constructing the kriging predictor in applied settings.

\subsection{Estimation of Parameter Surface}\label{sec:estimation}

To estimate the spatially varying parameters of the spatiotemporal process, we employ a nonparametric regression framework based on localised estimation. Since $Z_t(s_0)$ is unobserved at arbitrary locations $s_0$, classical approaches such as ordinary least squares are not directly applicable. Nevertheless, assuming smooth spatial variation of the model parameters (Assumption \ref{ass4}), nonparametric techniques allow their estimation at unobserved locations.

Inspired by local likelihood methods in nonparametric regression~\cite{tibshirani1987local} and the locally stationary framework developed by~\cite{Dahlhaus}, we adopt a weighted quasi-maximum likelihood approach. This method incorporates spatial weights, giving higher importance to observations near the prediction location, and allows for flexible estimation of spatially varying GARCH parameters while accounting for the complexities of spatiotemporal dependence.

Given the smoothness assumption in Assumption \ref{ass4}, we estimate the parameters at $s_0$ using a localised likelihood approach that weights observations from nearby locations more heavily. Specifically, we define the local log-likelihood
\begin{equation}
\mathcal{L}_{ T}\left(s_0, \Theta\right)= \frac{1}{m T'} \sum_{t=p+1}^T \sum_{u=1}^m W_b\left(s_0-s_u\right) I_{t} ( \Theta), 
\end{equation}
with
\begin{equation}
I_{t} ( \Theta) =\frac{1}{2} \left( \log \sigma_t^2(s_u) + \frac{Z_t(s_u)^2}{\sigma_t^2(s_u)} \right),
\end{equation}
where $\Theta = \left(\omega \left(s_0\right), \alpha_1(s_0)\ldots, \alpha_p\left(s_0\right) , \beta_1(s_0), \dots \beta_p(s_0)\right)$ and  ($ \alpha_{q+1}=0, \alpha_{q+2}=0 ,\dots, \alpha_{p}=0)$, $T'$ is the number of effective time points, i.e., $T' = T-p-1$.

The spatial weighting function is constructed from a kernel $K: \mathbb{R}^2 \rightarrow \mathbb{R}$ satisfying
\begin{equation}
    \int K(x) \mathrm{d} x=1, \quad W(s)=K(x) K(y) \text { and } W_b(.)=\frac{1}{b} W\left(\frac{.}{b}\right) . 
\end{equation}
The local parameter estimate $\hat{\Theta}(s_0)$ is then obtained as
\begin{equation}\label{eq:estiamtor}
\hat{\Theta}(s_0) = \arg \min_{\Theta} \mathcal{L}_{ T}\left(s_0, \Theta\right).
\end{equation}
This localised estimation procedure is computationally far more efficient than a full joint estimation of all spatial and temporal parameters (see Remark \ref{rmk:2}) and scales well to large spatial datasets.

The localised likelihood in \eqref{eq:estiamtor} provides pointwise estimates of the smoothly varying parameter functions at any location $s_0$. Having established the estimation framework, we next investigate its theoretical properties, focusing on consistency and asymptotic normality under the assumptions stated in Section~\ref{assumptions}.

\subsection{Properties of the Estimator}\label{sec:properties}

In this section, we analyse the theoretical properties of the weighted quasi-maximum likelihood estimator $\hat{\Theta}(s_0)$ for the spatiotemporal GARCH model. We work within a purely temporal asymptotic framework, treating the spatial grid of fixed size $m$ as fixed and letting the time series length $T \to \infty$. Under the previously stated regularity conditions, Theorem~\ref{consistency} establishes the consistency of $\hat{\Theta}(s_0)$, while Theorem~\ref{normality} derives its asymptotic normality.

\begin{theorem}[Consistency]\label{consistency}
Let \(\{Z_t(s_u):\,t=1,\dots,T,\;u=1,\dots,m\}\) be a spatiotemporal GARCH process satisfying Assumptions \ref{assumptions}.  Let
$$
\hat\Theta(s_0)
\;=\;
\arg\min_{\Theta}\;
\mathcal L_T\bigl(s_0,\Theta\bigr),
$$
be the local weighted quasi-likelihood estimator defined before. Suppose $m$ is fixed and $ b\rightarrow 0 \; \text{and}\; \;  bmT \rightarrow\;\infty \;  \text{as} \;  \;  T\rightarrow\infty,$
we have $$
\hat\Theta(s_0)\;\xrightarrow{\mathcal{P}}\;\Theta(s_0).
$$
\end{theorem}

The proof can be found in the Appendix. Having established the consistency, we now turn to the distributional behaviour of the estimator. The next theorem shows that $\hat\Theta(s_0)$ is asymptotically normal with a bias term of order $b^2$ due to the kernel smoothing.
 
\begin{theorem}[Asymptotic Normality]\label{normality}
Let $\{Z_t(s_u):\, t = 1,\dots,T,\; u = 1,\dots,m\}$ be a spatiotemporal GARCH process satisfying Assumptions \ref{assumptions}, and let $\hat\Theta(s_0)$ be the local quasi-likelihood estimator defined as
$$
\hat\Theta(s_0) = \arg \min_{\Theta } \mathcal{L}_T(s_0, \Theta),
$$
Then, we have the following:\\
$$
\sqrt{T'}\left(\hat{\Theta}(s_0)-\Theta(s_0)\right) + \sqrt{T'} \Sigma(s_0)^{-1}  \frac{1}{2m}\, b^2\, w^{(2)}\, \nabla_s^2 I \left(s_0, \Theta(s_0)\right)
\;\xrightarrow{d}\;
\mathcal{N}\left(0, \mu^2 \operatorname{Var}\left(\eta_t^2(s)\right) \Sigma(s_0)^{-1} \right),$$
where:
$$
\begin{aligned}
\mu^2 &= \sum_{u=1}^m \sum_{t=p+1}^{T} \frac{W_b(s_0 - s_u)^2}{2 m^2}, \\
\Sigma(s_0) &= \frac{1}{2} \mathbb{E}\left[ \frac{\nabla \widetilde{\sigma}_t^2(s)\, \nabla \widetilde{\sigma}_t^2(s)^\top}{\widetilde{\sigma}_t^4(s)} \right],\\   w^{(2)} &= \int_{} \| s-s_0\|^2 W( s-s_0)\, ds .
\end{aligned}$$
\end{theorem}

Together, Theorems~\ref{consistency} and \ref{normality} provide a rigorous justification for the localised likelihood estimation procedure introduced in the previous subsection. The second term in the asymptotic expansion of Theorem~\ref{normality} reflects a bias that arises from the deviation of the true process from its stationary approximation over the localised spatial segment. This bias depends on the curvature of the spatial parameter functions and quantifies how closely the process resembles a spatially stationary one in the neighbourhood of $s_0$.
It is therefore directly controlled by the bandwidth $b$, highlighting the bias–variance trade-off introduced by spatial kernel smoothing.

\section{Monte Carlo Simulation Study}\label{sec:MC}

We assess the performance of the proposed STGARCH estimator and predictor through a Monte Carlo simulation with \(MC = 100\) replications. In each replication, we independently and uniformly sample \(n_1 \in \{50, 100, 150\}\) training sites (\(S_1\)) and \(n_2 = 50\) prediction sites (\(S_2\)) from the unit square \([0,1]^2\). A fixed \(10 \times 10\) grid (\(S_0\)) is additionally used for evaluating parameter recovery on a dense spatial grid. We consider time series lengths \(T \in \{100, 200, 300\}\), where each spatially varying coefficient function \(\omega(s)\), \(\alpha(s)\), and \(\beta(s)\) is represented using a B-spline basis, with coefficients randomly drawn in each replication to ensure varying parameter surfaces across runs. The STGARCH(1,1) processes are simulated with spatially varying parameters and a spatially stationary Gaussian innovation process \(\eta_t(s) \sim \mathcal{N}(0, 1)\) following an exponential covariance function with range parameter \(\theta = 0.5\). The range parameter is estimated via the proposed weighted quasi-maximum likelihood method applied to residuals, yielding an average estimate of \(\hat{\theta} = 0.4992\) across replications (with $n_1 = 150$ and $T = 300$).

For each \(s_i \in S_1\), the STGARCH coefficient functions are estimated using weighted quasi-maximum likelihood with a uniform kernel and bandwidth \(b = n_1^{-1/4}\). The estimated parameters are then used to spatially predict unconditional volatilities at the locations in \(S_2\) via kriging, i.e., out-of-sample. The size of $S_2$ was chosen to be $n_2 = 50$. This is implemented as a two-step procedure: first, \(\omega(s)\), \(\alpha(s)\), and \(\beta(s)\) are estimated at each \(s_i\) via  the weighted localised QMLE, and second, the innovations \(\eta_t(s)\) and the derived \(\zeta_t(s)\) are computed for use in the kriging procedure. For reporting accuracy, we compute Monte Carlo (MC) averages of bias and RMSE of all parameters across simulation runs and spatial locations. That is, we define
\[
\mathrm{Bias}(\theta)
=\frac{1}{MC\,n_2}\sum_{r=1}^{MC}\sum_{s\in S_2}\bigl(\hat \theta^{(r)}(s)-\theta(s)\bigr),
\qquad
\mathrm{RMSE}(\theta)
=\left\{\frac{1}{MC\,n_2}\sum_{r=1}^{MC}\sum_{s\in S_2}\bigl(\hat \theta^{(r)}(s)-\theta(s)\bigr)^2\right\}^{1/2}.
\]
These metrics are reported separately for \(\theta \in\{\omega,\alpha,\beta\}\). Table~\ref{tab:bias-rmse-full} reports the average bias and RMSE for the estimated nonparametric functions \(\omega\), \(\alpha\), and \(\beta\) on \(S_2\), as well as for the predicted unconditional volatilities.


\begin{table}[ht]
\centering
\caption{Out-of-sample bias and RMSE of estimated nonparametric functions \(\omega(s)\), \(\alpha(s)\), \(\beta(s)\), and predicted unconditional volatilities at \(S_2\) over \(MC=100\) simulations.}\label{tab:bias-rmse-full}
\begin{tabular}{cc cc cc cc cc}
\hline
\multirow{2}{*}{\(n_1\)} & \multirow{2}{*}{\(T\)} 
& \multicolumn{2}{c}{\(\omega\)} 
& \multicolumn{2}{c}{\(\alpha\)} 
& \multicolumn{2}{c}{\(\beta\)} 
& \multicolumn{2}{c}{Volatility} \\
& & Bias & RMSE & Bias & RMSE & Bias & RMSE & Bias & RMSE \\
\hline
  & 100 & -0.0216 & 0.0784 & 0.0067 & 0.0774 & 0.0603 & 0.2126 & 0.0092 & 0.1205 \\
50  & 200 & -0.0308 & 0.0706 & 0.0042 & 0.0597 & 0.0827 & 0.1894 & 0.0042 & 0.1084 \\
  & 300 & -0.0348 & 0.0684 & 0.0047 & 0.0533 & 0.0922 & 0.1766 & 0.0037 & 0.1067 \\[.2cm]
 & 100 & -0.0177 & 0.0819 & 0.0073 & 0.0767 & 0.0447 & 0.2185 & 0.0020 & 0.1168 \\
100 & 200 & -0.0269 & 0.0718 & 0.0078 & 0.0605 & 0.0691 & 0.1850 & 0.0008 & 0.1096 \\
 & 300 & -0.0308 & 0.0651 & 0.0063 & 0.0527 & 0.0810 & 0.1671 & 0.0006 & 0.1058 \\[.2cm]
 & 100 & -0.0136 & 0.0819 & 0.0072 & 0.0820 & 0.0521 & 0.2182 & 0.0115 & 0.1231 \\
150 & 200 & -0.0240 & 0.0706 & 0.0056 & 0.0614 & 0.0698 & 0.1852 & 0.0002 & 0.1113 \\
 & 300 & -0.0270 & 0.0658 & 0.0058 & 0.0550 & 0.0771 & 0.1727 & 0.0004 & 0.1097 \\
\hline
\end{tabular}
\label{tab:bias-rmse-full}
\end{table}

To further illustrate the predictive performance of the proposed method, Figure~\ref{fig:spatial-volatility} compares the predicted and true unconditional volatility surfaces over the spatial grid \(S_0\). The background colour gradient represents the predicted values, with red dashed contours indicating the predicted level sets and black solid contours representing the true ones. The close alignment of these contours confirms that the estimator successfully recovers the underlying spatial volatility pattern. Complementing this spatial view, Figure~\ref{fig:temporal-volatility} presents the temporal evolution of the true and predicted conditional variances at three randomly selected out-of-sample locations in \(S_2\). The model closely follows both the timing and magnitude of volatility spikes, demonstrating its strong ability to capture dynamic features of the process across time and space.

\begin{figure}
  \centering
  \includegraphics[width=0.7\textwidth]{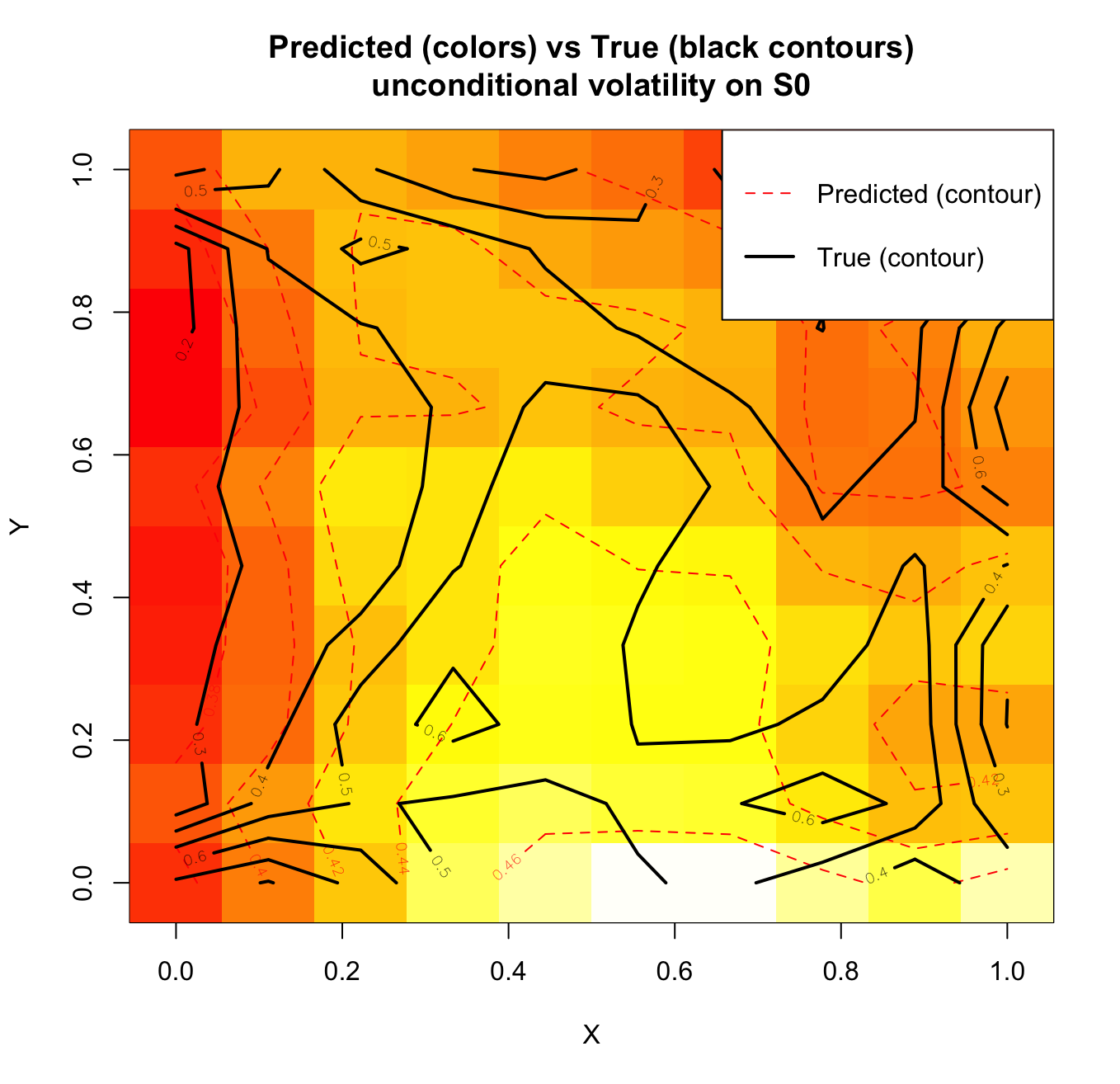}
  \caption{Predicted (background gradient, red dashed contours) and true (black solid contours) unconditional volatility surfaces on the spatial grid \(S_0\) for one replication.}
  \label{fig:spatial-volatility}
\end{figure}

\begin{figure}
  \centering
  \includegraphics[width=0.7\textwidth]{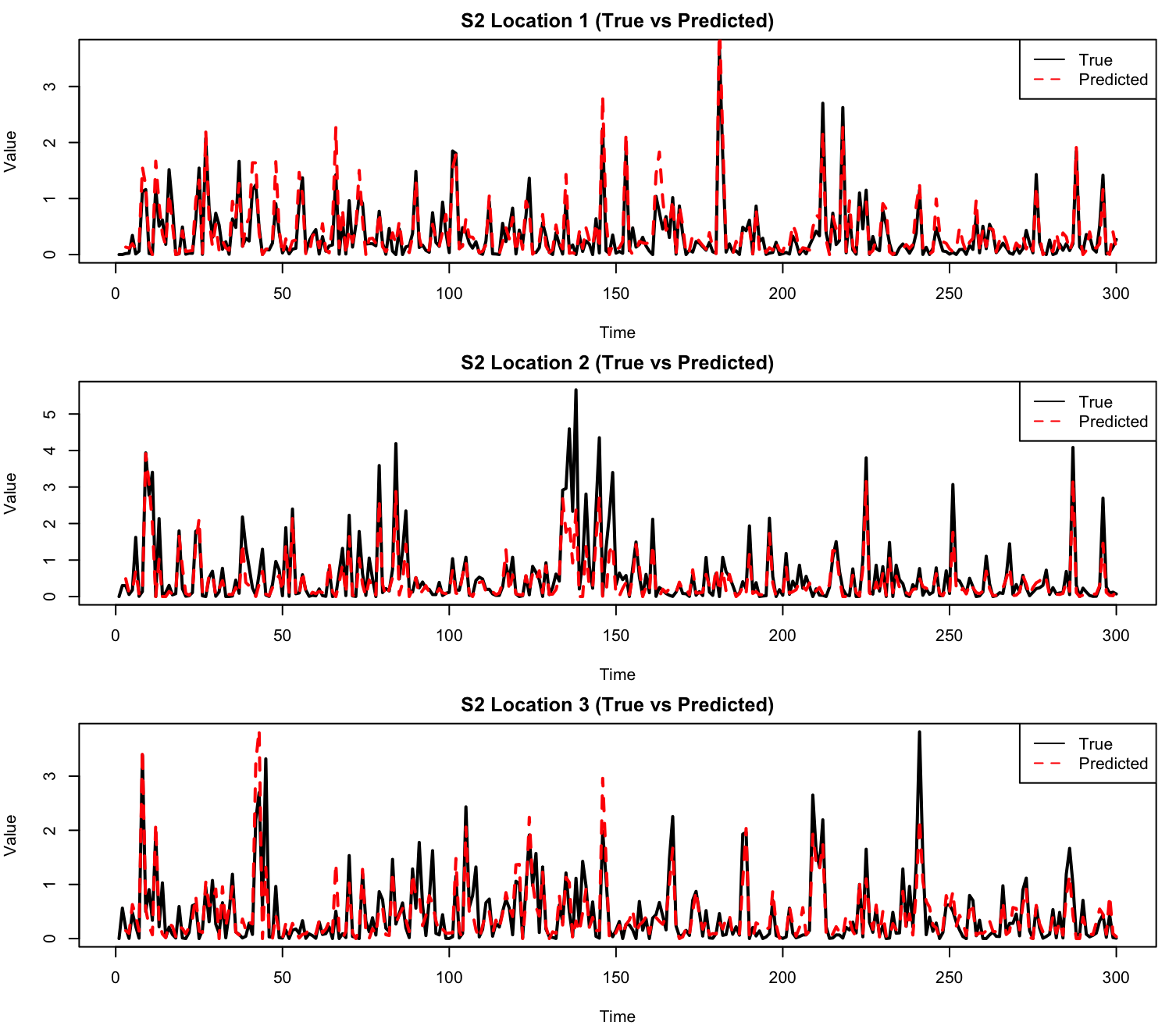}
  \caption{Temporal evolution of true (black) and predicted (red) conditional variances at three out-of-sample locations in \(S_2\). The predictions closely track both the timing and magnitude of volatility spikes.}
  \label{fig:temporal-volatility}
\end{figure}

Overall, the simulation results demonstrate that the proposed STGARCH estimation and prediction framework reliably recovers spatially varying parameters and accurately predicts both unconditional and conditional volatilities across a range of spatial sample sizes and time series lengths. Bias and RMSE values decrease as either \(n_1\) or \(T\) increases, reflecting the expected gains from richer spatial and temporal information. Spatial reconstruction closely matches the true volatility surfaces, and temporal predictions successfully capture both the timing and magnitude of volatility changes at unobserved locations. These findings confirm that the localised estimation-kriging approach offers both high predictive accuracy and computational efficiency, making it a practical tool for large-scale spatiotemporal volatility modelling.

\section{Volatility Prediction across Financial Networks}\label{sec:application}

To demonstrate the utility of the proposed spatiotemporal GARCH model with spatially varying coefficients and spatially correlated innovations, we apply it to the problem of volatility modelling in a network of financial assets. Our setting is non-standard from a geostatistical perspective: the spatial domain is not physical space but an abstract ``financial space'' constructed from firm-specific balance sheet features. This allows us to embed firms into a latent space in which spatial proximity reflects economic similarity, and where spatial dependence can be exploited for volatility prediction. 

The financial dataset consists of 50 firms traded on public markets, observed over 276 trading days in a single year. The full list of the 50 companies included in the analysis, spanning various sectors and geographic regions, is provided in Table~\ref{tab:stocks}. For each firm, 66 balance sheet variables were extracted. After removing variables with more than 5\% missing values, 44 features remained. Each feature was standardised across firms, and missing values were imputed with zero, effectively corresponding to the assumption that the mean of each feature is zero in the standardised space. This ensures that the feature space is centred and that the Euclidean distances are interpretable as deviations from the average firm. To provide a descriptive overview of the input data, we visualised the distributions of daily log-returns and absolute log-returns for all 50 firms (see Figure \ref{fig:data}). The boxplots in Figure 0a illustrate that log-returns are approximately symmetric around zero for most firms, but exhibit substantial heterogeneity in scale and tail behaviour. This variation motivates the use of firm-specific volatility dynamics in our modelling approach. Figure 0b displays the distribution of absolute log-returns, which serve as a proxy for realised volatility. Here, the differences across firms are even more pronounced, reflecting variation in baseline volatility levels, sensitivity to market conditions, or idiosyncratic shocks. These plots underscore the necessity of a model that allows for spatially varying volatility structures, as pursued in our spatiotemporal GARCH specification.

\begin{figure}
\centering
\includegraphics[width = 0.45\textwidth]{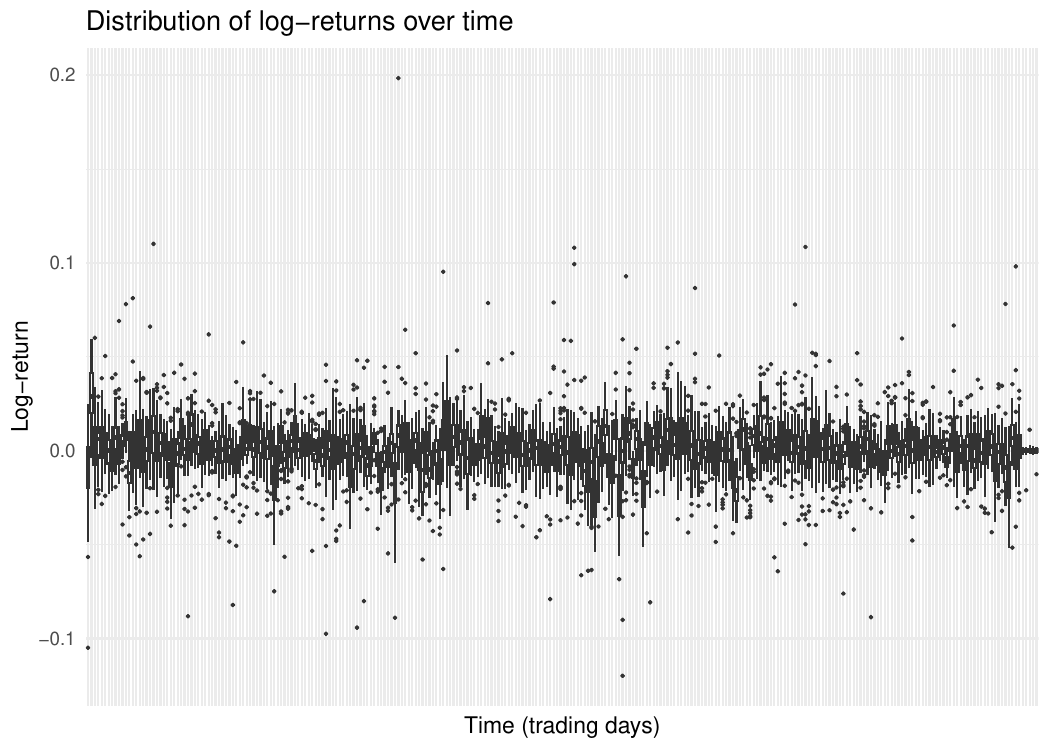}
\includegraphics[width = 0.45\textwidth]{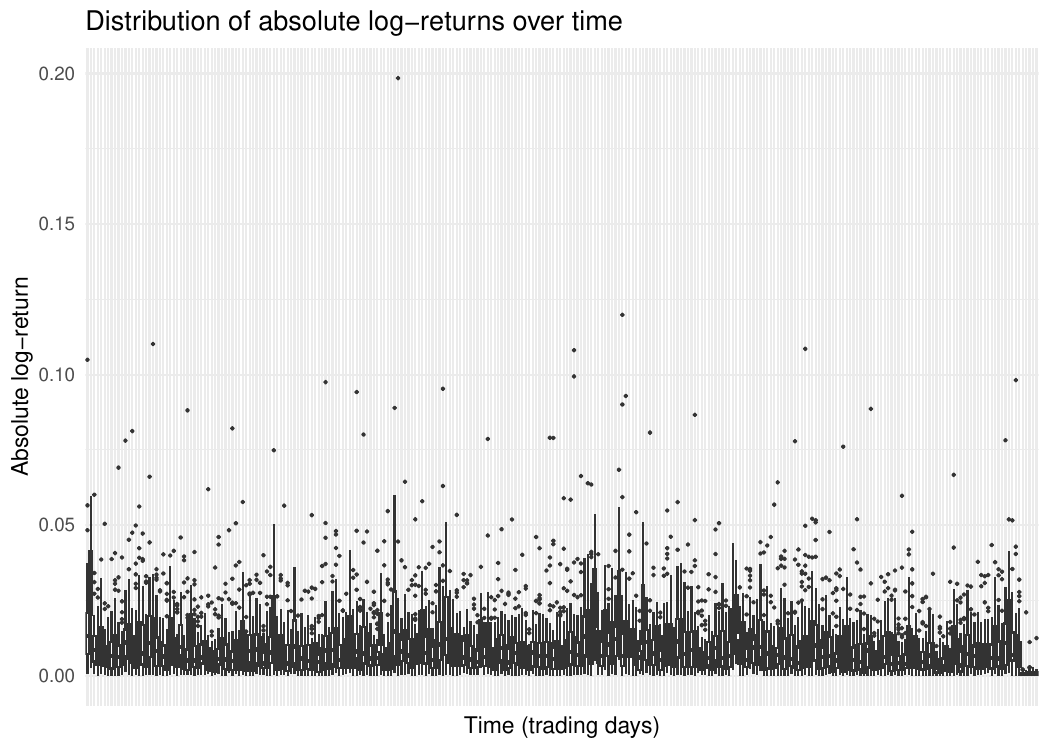}\\
\includegraphics[width = 0.45\textwidth]{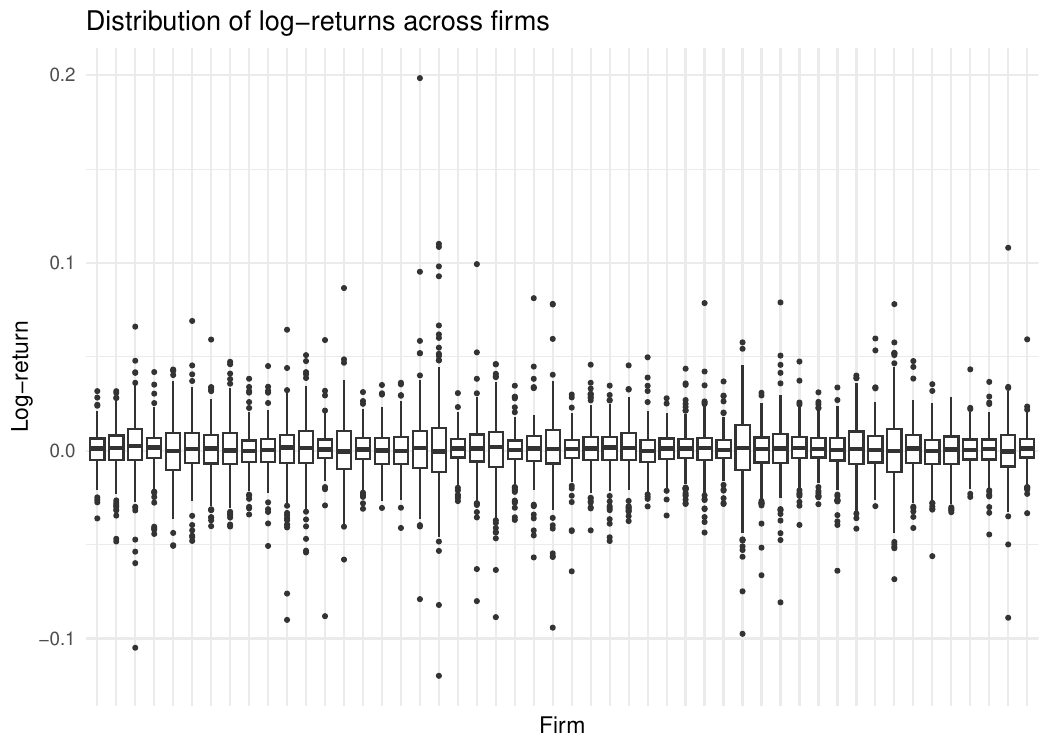}
\includegraphics[width = 0.45\textwidth]{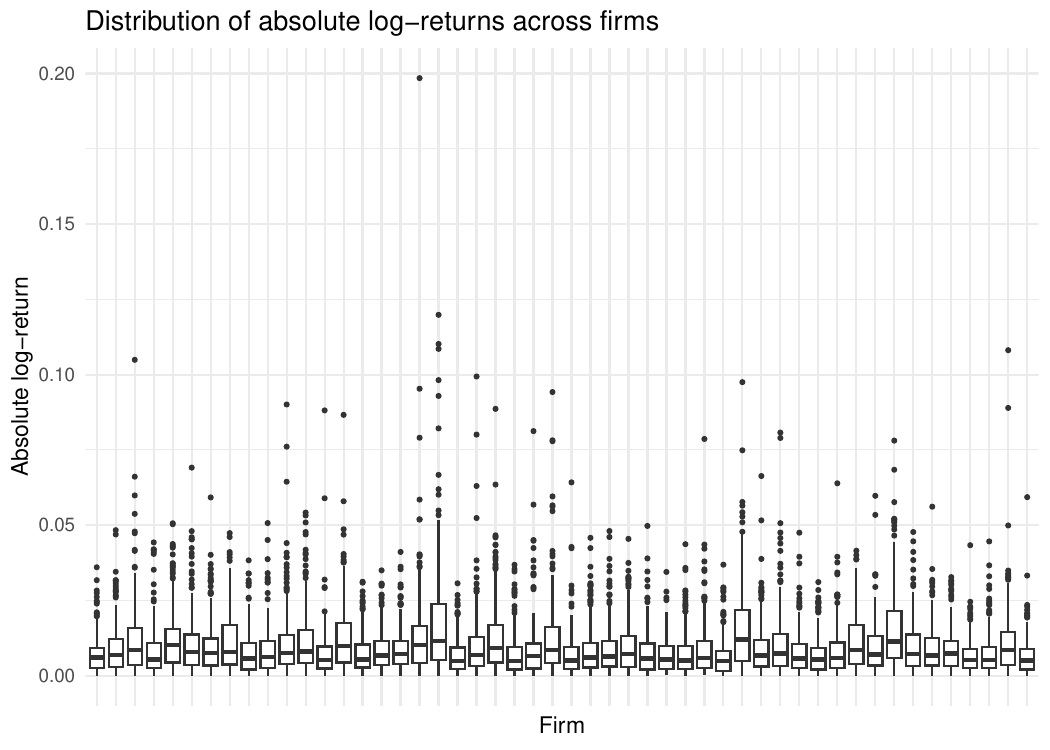}
\caption{Boxplots of log-returns (left panels) and absolute log-returns (right panels), displayed over time (top row) and across firms (bottom row). The absolute log-returns serve as a proxy for volatility. The top row visualises temporal variation and clustering in volatility, while the bottom row reveals differences in return distributions across firms.}\label{fig:data}	
\end{figure}

To reconstruct the financial network structure, we followed the strategy proposed by \citep{fulle2024spatial}, where firms are embedded in a low-dimensional proxy space based on selected balance sheet positions. Since no true spatial network is observable in this context, the model’s performance depends heavily on how well the constructed space reflects economic relationships relevant to volatility transmission. Instead of selecting a single financial proxy a priori, we systematically explored all pairwise combinations of the 44 features to define candidate two-dimensional spaces. Each of the resulting 946 spaces defines a distinct notion of ``distance'' among firms. For each space, we fitted the spatiotemporal GARCH model and evaluated its out-of-sample performance in predicting volatility.

Volatility predictions were assessed using five-fold cross-validation. In each fold, the model parameters were estimated on the full time series of a subset of firms and used to predict volatility on the held-out firms (i.e., we left the entire time series out). While this random allocation of firms into folds does not reflect realistic market clustering, and may lead to a mild overestimation of predictive performance \cite{otto2024review_CV}, block cross-validation is not feasible in this setting as the underlying network structure is unknown. Moreover, for a fair comparison, we kept the fold assignments constant for all 946 pairwise proxy spaces. 

Figure \ref{fig:rmse} displays the ranked RMSE values across all 946 proxy spaces. The RMSE is computed between the predicted volatility $\hat{z}^2_t(s)$ and the absolute observed log-return $|y_t(s)|$, providing a direct measure of how well the conditional standard deviations reflect the magnitude of return innovations. Among all candidate spaces, the best prediction performance was obtained for the space defined by the variables \emph{other current liabilities} and \emph{long-term debt}, with an RMSE of approximately 0.0112. This proxy space reflects the structure of firms’ liabilities—both short- and long-term—and may be interpreted as a financial exposure dimension. Economically, it is plausible that firms with similar liability profiles are jointly sensitive to changes in credit conditions or interest rates, which in turn influence their volatility responses to market shocks.

Other high-performing spaces include combinations such as \emph{total cash from operating activities} with \emph{other current liabilities}, and \emph{net income} with \emph{total cash flows from investing activities}. These pairings suggest that profitability and internal cash generation are also relevant for explaining co-movement in volatility. Spaces involving net borrowings and various forms of net income point to the role of financing strategy and retained earnings in shaping a firm’s exposure to systemic volatility.

While several of these spaces result in near-identical RMSE values, we will, for the sake of interpretability and visual clarity, focus our analysis on the optimal two-dimensional space defined by other current liabilities and long-term debt. In principle, the analysis could be extended to three- or four-dimensional spaces, potentially lowering the RMSE further and capturing additional economic structure. However, such extensions come at the cost of reduced interpretability and visualisation clarity, especially when aiming to represent spatial variation in parameter surfaces and volatility fields.

\begin{figure}
\centering
\includegraphics[width = 0.8\textwidth]{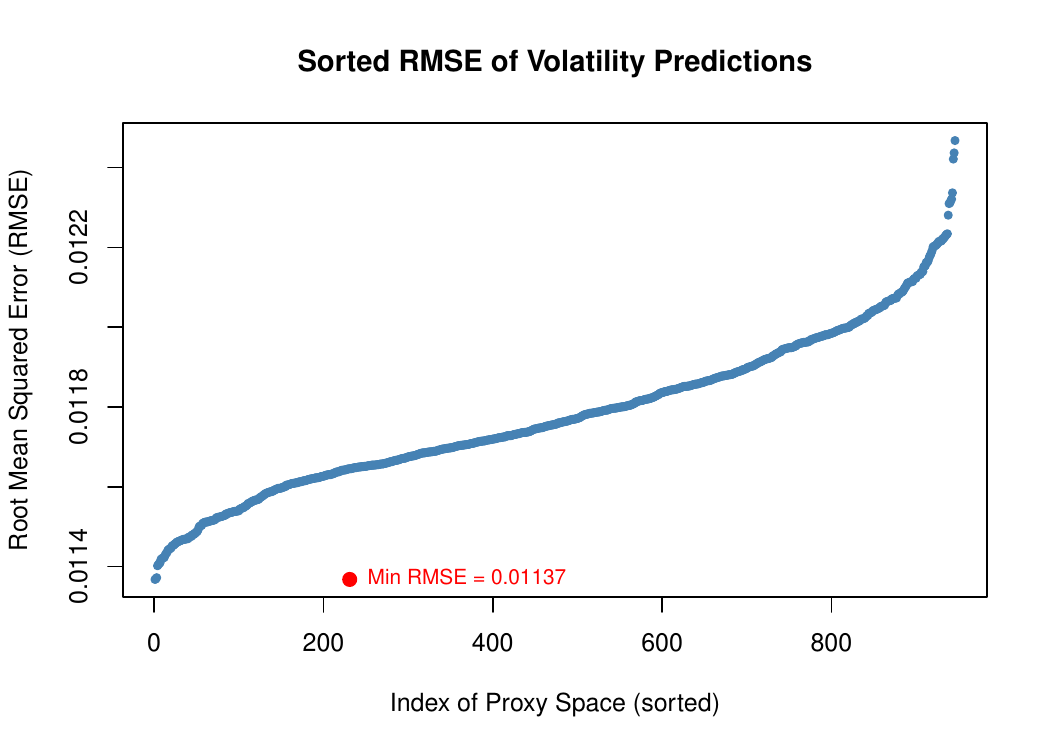}
\caption{Sorted root mean squared errors (RMSE) from the volatility prediction cross-validation study across all 946 candidate proxy spaces derived from pairwise combinations of balance sheet variables. Each point corresponds to the average RMSE obtained from a five-fold cross-validation, where each fold leaves out the full time series of a subset of firms. The narrow spread around the optimal value indicates that several proxy spaces yield comparably good predictive performance, highlighting robustness with respect to the choice of the latent spatial configuration.}\label{fig:rmse}	
\end{figure}

Figure \ref{fig:volatility_combined} provides a joint representation of the temporal and spatial structure of the predicted volatilities obtained from cross validation study. The left panel shows the predicted volatility trajectories over time, with grey lines representing individual firm-level volatilities and the bold black line denoting the cross-sectional median. This illustrates both market-wide and firm-specific volatility dynamics across the year 2019, with evident spikes around mid-year and increased volatility clustering. The right panel displays the firms in the optimal two-dimensional proxy space spanned by \emph{other current liabilities} and \emph{long-term debt}. This reveals substantial spatial heterogeneity in the volatility dynamics: firms with similar financial characteristics exhibit similar volatility levels and uncertainty, supporting the relevance of the proxy space for modelling and predicting financial risk.

\begin{figure}
  \centering
  \includegraphics[width=0.52\textwidth]{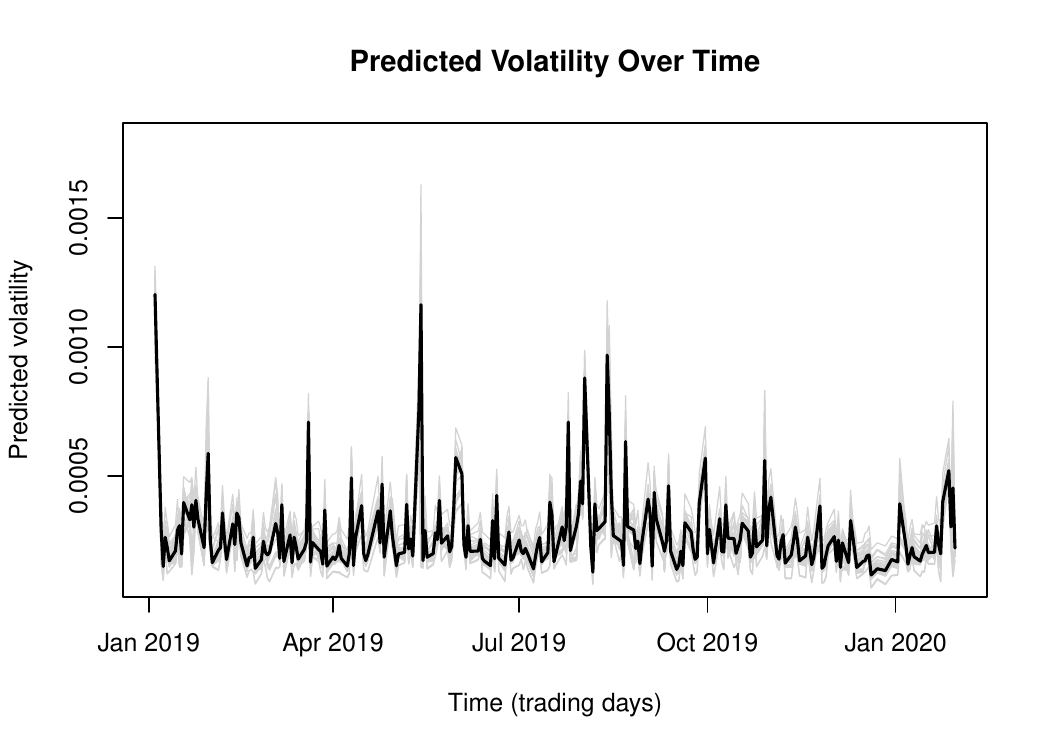} \hfill
  \includegraphics[width=0.38\textwidth]{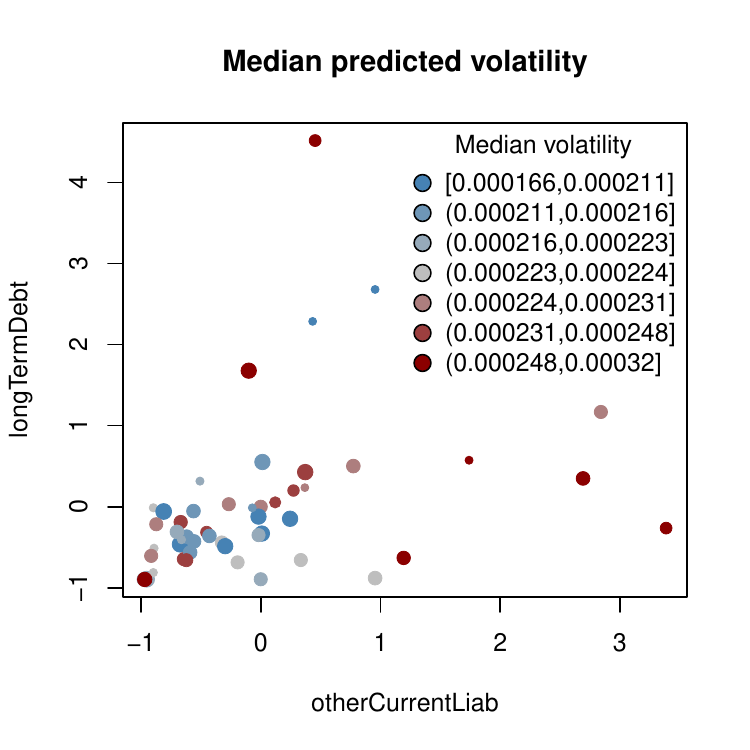}
  \caption{Predicted volatilities from the heterogeneous spatiotemporal GARCH model in the optimal 2-dimensional balance sheet space. The left panel shows the temporal evolution: the black curve indicates the median volatility across firms, and the grey curves represent individual firm-level predictions. The right panel depicts the firm positions in the proxy space; point colour denotes median predicted volatility and point size encodes its standard deviation (of the estimated volatilities).}
  \label{fig:volatility_combined}
\end{figure}

To analyse the selected space in more detail, we re-estimated the model using all firms and all 276 days. The innovation covariance parameters were estimated via maximum likelihood, yielding $\hat{\theta} = 80.3695$, $\hat{\sigma}^2 = 0.0573$, and $\hat{\tau}^2 = 0.17173$. These values suggest that the spatial correlation in the innovation process decays slowly. Indeed, the maximum distance among firms in the selected space was only 5.6. The large estimate of $\theta$ implies that nearly all firms are subject to common shocks, with little attenuation based on financial distance. This is further supported by the relatively large nugget effect $\tau^2$, which captures idiosyncratic variation not explained by the spatial process. While the variance parameters were statistically significant, the range parameter had a relatively large standard error, yielding a z-score of 1.419.

\begin{table}
\centering
\caption{List of the 50 companies included in the analysis.}
\begin{tabular}{ll}
\hline
\multicolumn{2}{c}{\textbf{Companies}} \\
\hline
Allianz        & JPM         \\
Abbott         & Merck       \\
Apple          & Microsoft   \\
AT\&T          & MUFG        \\
Santander      & Nestlé      \\
Bank of America & Novartis   \\
BHP            & Oracle      \\
BNP Paribas    & Pepsi       \\
BP             & Petrobras   \\
Chevron        & Pfizer      \\
Cisco          & PMI         \\
Citigroup      & Procter \& Gamble \\
Coca-Cola      & Roche       \\
ConocoPhillips & Shell       \\
E.ON           & Samsung     \\
Eni            & Sanofi      \\
ExxonMobil     & Schlumberger \\
Gazprom        & Siemens     \\
General Electric & Telefónica \\
GlaxoSmithKline & Total      \\
Google         & Toyota      \\
HP             & Verizon     \\
HSBC           & Vodafone    \\
IBM            & Walmart     \\
Intel          & Johnson \& Johnson \\
\hline
\end{tabular}
\label{tab:stocks}
\end{table}

Figure \ref{fig:parameter_surfaces} displays the estimated spatial surfaces of the three GARCH parameters, $\beta(s)$, $\alpha(s)$, and $\omega(s)$. To ensure stable and reliable estimation, we restrict the analysis to the convex hull of the financial proxy space. This is crucial because local likelihood estimators are known to perform poorly near the boundary, where neighbourhoods become asymmetric and effective sample sizes are insufficient.

The resulting parameter surfaces reveal substantial spatial heterogeneity in volatility dynamics across firms. The spatial variation of the parameters reveals distinct patterns across the financial proxy space. The persistence parameter $\beta(s)$ attains its highest values in the north-eastern and northern regions, corresponding to firms with high long-term debt and moderate to high current liabilities. This suggests that highly leveraged firms exhibit more persistent volatility dynamics, possibly due to slower adjustments in their financial risk profiles. In contrast, firms in the south-western region (characterised by low debt and liabilities) show lower persistence, indicating that shocks dissipate more rapidly. The innovation parameter $\alpha(s)$ is largest in the south-western region, highlighting that more conservatively financed firms respond more strongly to recent return shocks. Lastly, the unconditional volatility level $\omega(s)$ is elevated in the northern part of the space, again associated with high-debt firms, reflecting a higher baseline risk level. These spatial patterns offer economically meaningful insights into the relationship between firm-specific financial characteristics and the temporal dynamics of their volatility. Such heterogeneity in volatility dynamics is critical for modelling systemic risk in financial networks and may guide the development of portfolio allocation strategies and stress-testing frameworks that explicitly account for firm-specific balance sheet configurations.

\begin{figure}
\centering
\includegraphics[width=0.32\textwidth]{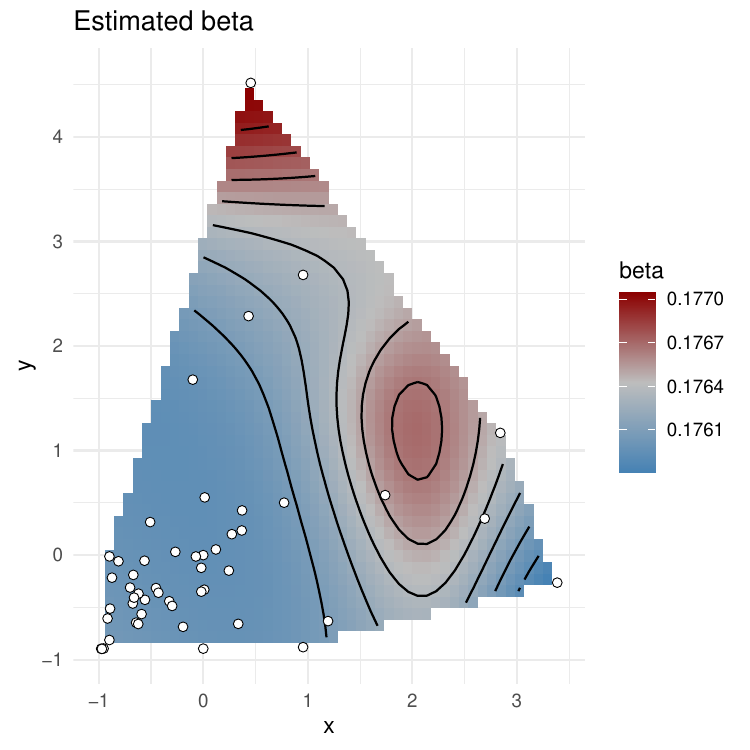}
\includegraphics[width=0.32\textwidth]{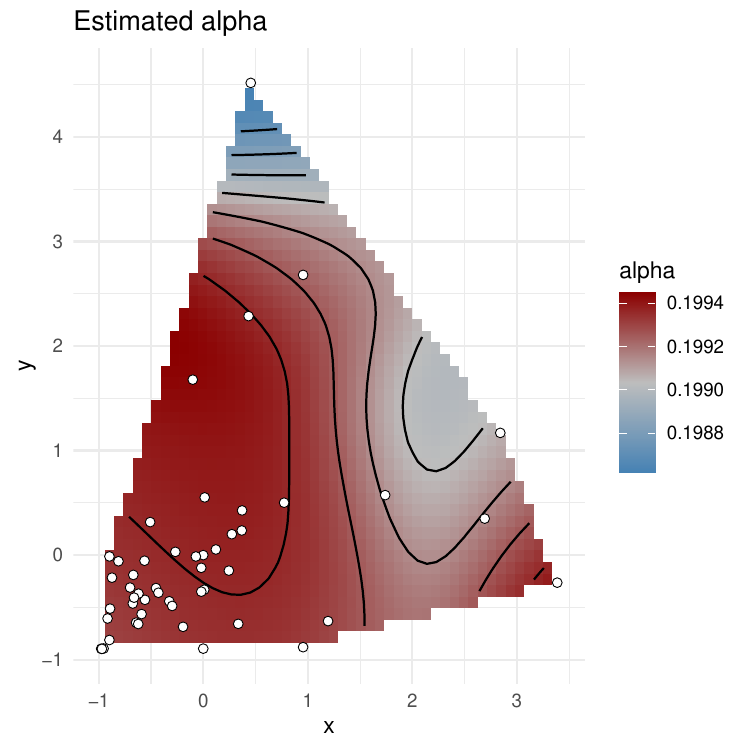}
\includegraphics[width=0.32\textwidth]{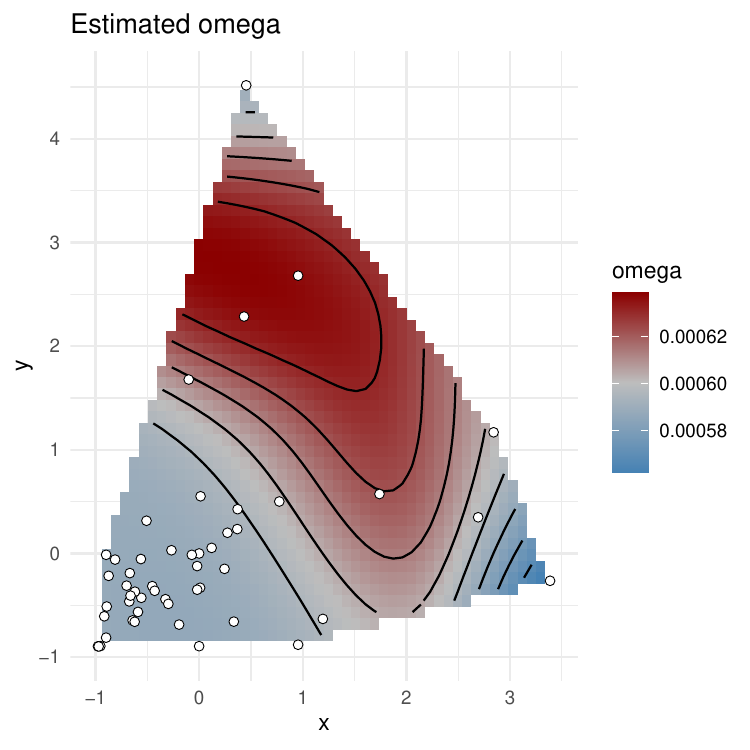}
\caption{Estimated spatial surfaces of the three GARCH parameters: $\beta$ (left), $\alpha$ (centre), and $\omega$ (right), restricted to the convex hull of the firm locations.}
\label{fig:parameter_surfaces}
\end{figure}

\section{Conclusion}\label{sec:conclusion}

We introduced a heterogeneous spatiotemporal GARCH framework on continuous spatial domains that permits spatially varying volatility parameters and contemporaneous cross‐sectional dependence through a geostatistical innovation field. We developed a localised quasi–likelihood estimator that yields smooth parameter surfaces and a kriging scheme for predicting volatility at unobserved locations. Theoretical results established consistency and asymptotic normality under a fixed–$m$, $T\to\infty$ regime. In simulations, the method recovered parameter surfaces with low bias and RMSE and produced accurate spatial and temporal volatility predictions. In an empirical study on a financial proxy space spanned by balance‐sheet variables, the model captured economically interpretable spatial heterogeneity in baseline risk, shock sensitivity, and persistence, and delivered competitive predictive performance while remaining computationally efficient via a two–stage estimation–kriging procedure.

Several limitations suggest natural extensions. The Gaussian–process assumption for the innovation field is convenient but not essential; accommodating heavy–tailed or skewed innovations would improve robustness to extreme events, which might be particularly useful for financial data. The fixed–domain asymptotics with increasing time simplify analysis but leave open theory for double asymptotics with $m\to\infty$ and $T\to\infty$, data–driven bandwidth selection, and boundary correction. Another promising direction for future research is to allow coefficients to vary in both space and time, i.e., $\omega(s,t)$, $\alpha(s,t)$, and $\beta(s,t)$, enabling local regime changes and structural breaks to be captured explicitly; this, in turn, calls for a space–time kriging framework for volatility that propagates uncertainty across both dimensions. Extending the spatial domain beyond Euclidean settings to networks and manifolds with covariance functions valid on graphs or spheres would broaden applicability to financial, transportation, and global climate systems. On the computational side, joint estimation with sparsity or low‐rank structure, or adaptive bandwidths. Finally, for the empirical application, model selection for the proxy space, multivariate co–volatility extensions, and mechanisms for time‐varying network topology offer promising avenues for capturing evolving interdependencies in financial markets and beyond.

\bibliographystyle{apalike} 

\appendix
\section{Appendix}
 \subsection{Proofs of Estimator Properties}\label{A}
In this appendix, we collect the detailed proofs of the main asymptotic results stated in Section \ref{sec:properties} .  In particular, we first establish the consistency of the local weighted quasi‑maximum likelihood estimator (Theorem \ref{consistency}), and then derive its asymptotic normality (Theorem \ref{normality}).

\begin{proof}\emph{Proof of Theorem \ref{consistency}.}

We recall the weighted quasi‑maximum likelihood estimator at \(s_{0}\) by
$$
\hat{\Theta}(s_{0}) \;=\; \arg\min_{\Theta\in\Omega}\mathcal{L}_{T}(s_{0},\Theta),
$$
where
$$
\mathcal{L}_{T}(s_{0},\Theta)
\;=\;
\frac{1}{m\,T'} \sum_{u=1}^{m}\sum_{t=p+1}^{T}
W_{b}(s_{0}-s_{u})\,I_{t}(\Theta),
\quad
T' = T - p - 1,
\quad
W_{b}(s) = b^{-2}K(x/b)\,K(y/b),
$$
and
$$
I_{t}(\Theta)
\;=\;
\frac{1}{2}\Bigl(
\log\sigma_{t}^{2}(s_{u})
+
\tfrac{Z_{t}^{2}(s_{u})}{\sigma_{t}^{2}(s_{u})}
\Bigr).
$$
Moreover, let $\Theta$ be the entire parameter surface across all $s \in \mathcal{D}$.

We also define the stationary–approximate criterion
$$
\widetilde{\mathcal{L}}_{T}(s_{0},\Theta)
\;=\;
\frac{1}{m\,T'} \sum_{u=1}^{m}\sum_{t=p+1}^{T}
W_{b}(s_{0}-s_{u})\,\widetilde{I}_{t}(s_{0};\Theta),
$$
where
$$
\widetilde{I}_{t}(s_{0};\Theta)
\;=\;
\frac{1}{2}\Bigl\{
\log\bigl[\widetilde{\sigma}_{t}^{2}(s)\bigr]
+
\tfrac{\widetilde{Z}_{t}(s)^{2}}
     {\widetilde{\sigma}_{t}^{2}(s)}
\Bigr\},
$$
and \(\widetilde{Z}_{t}(s)\) is locally stationary \ref{spa}.

\medskip\noindent\textbf{Step 1.}  
Set
$$
\mathcal{B}_{T}(s_{0},\Theta)
\;=\;
\mathcal{L}_{T}(s_{0},\Theta)
-
\widetilde{\mathcal{L}}_{T}(s_{0},\Theta).
$$
We show
$$
\sup_{\Theta}\bigl|\mathcal{B}_{T}(s_{0},\Theta)\bigr|
\;\xrightarrow{\mathcal{P}}\;0.
$$
Indeed,
$$
\bigl|\mathcal{B}_{T}\bigr|
\le
\frac{1}{2m\,T'}\sum_{u,t}
W_{b}(s_{0}-s_{u})
\bigl|I_{t}(\Theta)-\widetilde{I}_{t}(s_{0};\Theta)\bigr|
\;\le\;
I_{1}+I_{2},
$$
with
$$
I_{1}
=
\frac{1}{2m\,T'}\sum_{u,t}
W_{b}(s_{0}-s_{u})
\bigl|\log\sigma_{t}^{2}(s_{u})
-\log\widetilde{\sigma}_{t}^{2}(s)\bigr|,
$$
$$
I_{2}
=
\frac{1}{2m\,T'}\sum_{u,t}
W_{b}(s_{0}-s_{u})
\Bigl|
\tfrac{Z_{t}^{2}(s_{u})}{\sigma_{t}^{2}(s_{u})}
-
\tfrac{\widetilde{Z}_{t}^{2}(s)}
      {\widetilde{\sigma}_{t}^{2}(s)}
\Bigr|.
$$
Under Assumption \ref{ass2}, stationarity, and ergodicity, and by the deviation bounds in Theorem \ref{spa} and Corollary \ref{cor}, both \(I_{1}\) and \(I_{2}\) are \(o(1)\).

\medskip\noindent\textbf{Step 2.}  
Let
$$
L(s_{0},\Theta)
=\mathbb{E}\bigl[\widetilde{I}_{t}(s_{0};\Theta)\bigr].
$$
By the ergodic theorem in \(t\) and the fact that \(\sum_{u}W_{b}(s_{0}-s_{u}) \approx 1\), we have for each fixed \(\Theta\),
$$
\widetilde{\mathcal{L}}_{T}(s_{0},\Theta)
\;\xrightarrow{\mathcal{P}}\;
L(s_{0},\Theta).
$$
A standard equicontinuity argument on the compact \(\Omega\) then yields
$$
\sup_{\Theta}
\bigl|\widetilde{\mathcal{L}}_{T}(s_{0},\Theta)
- L(s_{0},\Theta)\bigr|
\;\xrightarrow{\mathcal{P}}\;
0.
$$

\medskip\noindent\textbf{Step 3.}  
Combining Steps 1–2 gives
$$
\sup_{\Theta}
\bigl|\mathcal{L}_{T}(s_{0},\Theta)
- L(s_{0},\Theta)\bigr|
\;\xrightarrow{\mathcal{P}}\;
0.
$$

\medskip\noindent\textbf{Step 4.}  
Let \(\Theta_{s_{0}}=\arg\min_{\Theta}L(s_{0},\Theta)\), which is unique by Lemma 5.5 of \cite{Horv}.  
By definition of \(\hat{\Theta}\) and uniform convergence,
$$
\mathcal{L}_{T}(s_{0},\hat{\Theta})
\;\le\;
\mathcal{L}_{T}(s_{0},\Theta_{s_{0}})
\;\xrightarrow{\mathcal{P}}\;
L(s_{0},\Theta_{s_{0}})
\;\le\;
L(s_{0},\hat{\Theta}).
$$
Hence \(\hat{\Theta}_{s_0}\xrightarrow{\mathcal{P}}\Theta_{s_0}\).

\end{proof}

\begin{proof}\emph{Proof of Theorem \ref{normality}.}

We use Taylor expansion around $\Theta_{s_0}$:

\begin{equation}\label{eq5}
\nabla \mathcal{L}\left(s_0, \hat{\Theta}\right)
=
\nabla \mathcal{L}_{T}\left(s_0, \Theta_{s_0}\right)
+
\nabla^2 \mathcal{L}_{T}\left(s_0, \Theta_{s_0}\right)\left(\hat{\Theta} - \Theta_{s_0}\right)
\end{equation}

Since $\nabla \mathcal{L}\left(s_0, \hat{\Theta}\right) = 0$, we get
$$
-\nabla \mathcal{L}\left(s_0, \Theta_{s_0}\right)
=
\nabla^2 \mathcal{L}\left(s_0, \Theta_{s_0}\right)\left(\hat{\Theta} - \Theta_{s_0}\right).
$$

We also have
$$
\mathscr{B}_t = \mathcal{L}\left(s_0, \Theta_{s_0}\right) - \tilde{\mathcal{L}}\left(s_0, \Theta_{s_0}\right),
$$
which implies
$$
\nabla \mathscr{B}_t = \nabla \mathcal{L}\left(s_0, \Theta_{s_0}\right) - \nabla \tilde{\mathcal{L}}\left(s_0, \Theta_{s_0}\right),
$$
and hence
$$
\nabla \mathcal{L}\left(s_0, \Theta_{s_0}\right)
=
\nabla \mathscr{B}_t + \nabla \tilde{\mathcal{L}}\left(s_0, \Theta_{s_0}\right).
$$

Replacing this in Equation \ref{eq5}, we obtain:
$$
- \nabla \mathscr{B}_t - \nabla \tilde{\mathcal{L}}\left(s_0, \Theta_{s_0}\right)
=
\nabla^2 \mathcal{L}\left(s_0, \Theta_{s_0}\right)\left(\hat{\Theta} - \Theta_{s_0}\right),
$$
or equivalently,
$$
\nabla^2 \mathcal{L}\left(s_0, \Theta_{s_0}\right)\left(\hat{\Theta} - \Theta_{s_0}\right)
+
\nabla \mathscr{B}_t
=
- \nabla \tilde{\mathcal{L}}\left(s_0, \Theta_{s_0}\right).
$$

We use $\nabla \tilde{\mathcal{L}}\left(s_0, \Theta_{s_0}\right)$ to show asymptotic normality:
\begin{equation}
\nabla \tilde{\mathcal{L}}\left(s_0, \Theta_{s_0}\right)
=
\frac{1}{m T'} \sum_{u=1}^m \sum_{t=p+1}^T W_b(s_0 - s_u)\, \nabla I_t(s_0,\Theta),
\end{equation}
where
$$
\nabla I_t(s_0,\Theta)
=
\frac{1}{2} \mathbb{E} \left[
\frac{- \nabla \widetilde{\sigma}_t^2(s)}{\widetilde{\sigma}_t^2(s)}
-
\frac{\eta_t^2(s)\, \nabla \widetilde{\sigma}_t^2(s)}{\widetilde{\sigma}_t^2(s)}
\right].
$$

$\nabla I_t(s_0,\Theta)$ is a martingale difference, which implies that $\nabla \tilde{\mathcal{L}}\left(s_0, \Theta_{s_0}\right)$ is a weighted sum of martingale differences. Using the central limit theorem and verifying the conditions of Corollary 3.1 \cite{hall2014martingale}, we conclude that
$$
\sqrt{T'}\, \nabla \tilde{\mathcal{L}}\left(s_0, \Theta_{s_0}\right)
\xrightarrow{\mathcal{D}}
\mathcal{N}\left(0,\, \mu^2\, \operatorname{var}\left[\eta_t^2(s)\right]\, \Sigma(s_0)\right),
$$
where
$$
\mu^2 = \sum_{u=1}^m \sum_{t=p+1}^T \frac{W_b(s_0 - s_u)^2}{2m^2},
\qquad
\Sigma(s_0) = \frac{1}{2} \mathbb{E} \left[
\frac{\nabla \widetilde{\sigma}_t^2(s)\, \nabla \widetilde{\sigma}_t^2(s)^\top}
{\widetilde{\sigma}_t^4(s)}
\right].
$$

As a result, we have
$$
\sqrt{T'}\left(\hat{\Theta} - \Theta_{s_0}\right)
+
\sqrt{T'}\, \Sigma(s_0)^{-1} \nabla \mathscr{B}_t
\xrightarrow{\mathcal{D}}
\mathcal{N}\left(0,\, \mu^2\, \operatorname{var}\left(\eta_t^2(s)\right)\, \Sigma(s_0)^{-1}\right).
$$

Now, consider
$$
\nabla \mathscr{B}_t
=
\frac{1}{m T'} \sum_{u=1}^m \sum_{t=p+1}^T W_b(s_0 - s_u)
\left[
\nabla I_t(\Theta) - \nabla I_t(s_0; \Theta)
\right].
$$

Using a Taylor expansion of $\nabla I_t(\Theta)$ around $s_0$, we get:
$$
\nabla I_t\left(\Theta\right)
=
\nabla I_t\left(s_0, \Theta_{s_0}\right)
+
\nabla I_t\left(s_0, \Theta_{s_0}\right)(s - s_0)
+
\frac{1}{2}(s - s_0)^2\, \nabla I\left(s_0, \Theta_{s_0}\right).
$$

Hence,
$$
\nabla \mathscr{B}_t
=
\frac{1}{m T'} \sum_{u=1}^m \sum_{t=p+1}^T W_b(s_0 - s_u)
\left[
\nabla I_t\left(s_0, \Theta_{s_0}\right)(s - s_0)
+
\frac{1}{2}(s - s_0)^2\, \nabla I\left(s_0, \Theta_{s_0}\right)
\right],
$$
where
$$
\mathbb{E}[\nabla \mathscr{B}_t]
=
\frac{1}{2m}\, b^2\, w^{(2)}\, \nabla I\left(s_0, \Theta_{s_0}\right),
$$
and
$$
w^{(2)} = \int \|s - s_0\|^2 W(s - s_0)\, ds,
$$
is the second moment of the kernel, while
$$
\operatorname{Var}[\nabla \mathscr{B}_t]
= O\left( \frac{b^4}{m T'} \right).
$$

Hence, we obtain:
$$
\sqrt{T'}\, \Sigma(s_0)^{-1} \nabla \mathscr{B}_t
=
\sqrt{T'}\, \Sigma(s_0)^{-1} \left(\frac{1}{2m} b^2 w^{(2)}\, \nabla I\left(s_0, \Theta_{s_0}\right)\right)
+ o(1).
$$

As a result, we conclude:
$$
\sqrt{T'}\left(\hat{\Theta} - \Theta_{s_0}\right)
+
\sqrt{T'}\, \Sigma(s_0)^{-1} \left( \frac{1}{2m} b^2 w^{(2)}\, \nabla I\left(s_0, \Theta_{s_0}\right) \right)
\xrightarrow{\mathcal{D}}
\mathcal{N}\left(0,\, \mu^2\, \operatorname{var}\left(\eta_t^2(s)\right)\, \Sigma(s_0)^{-1}\right).
$$
\end{proof}

\subsection{ The Approximated Spatially Stationary Process}\label{B}

In this section, we investigate the probabilistic properties of the STGARCH process. Specifically, we demonstrate that the spatially nonstationary process $ Z_t^2(s) $ can be locally approximated by a spatially stationary process $\widetilde{Z}^2_{t}(s) $ in the neighborhood of any fixed spatial location $ s_0 $.
This approximation plays a critical role in establishing the asymptotic properties of the estimator presented in Appendix \ref{A}.
We begin with the following lemma, which is essential for proving that the process $Z^2_t(s)$can be locally approximated by a spatially stationary process.\\

\begin{lemmaB}\label{lem}
Let $Z_t(s)$ be  STGARCH(p,q) model  and satisfies the previous Assumptions \ref{ass1} to \ref{ass5}. Suppose that for some $1-\gamma<\lambda<1$: $$\nu_k(s) =  [A(s)^k]_{1,1}.$$  We assume for all k and $ C > 0$
\begin{align*}
&(1) \quad \sup_{s} |\nu_k(s)| \leq C \lambda^k. \\
&(2) \quad \sup_s \mathbb{E}(|Z_t^2(s)|^{4 + v}) < \infty. \\
&(3) \quad \sup_{s,v} |\nu_k(s) - \nu_k(v)| \leq C k \lambda^{k-1} \|s - v\|_{\infty}. \\
&(4) \quad \text{The second partial derivatives of } \nu_k(s)  \text{ exist, with} \\
&\quad \left| \frac{\partial \nu_k(x,y)}{\partial x } \right| \leq C k \lambda^{k-1}, \quad 
\left| \frac{\partial \nu_k(x,y)}{\partial y } \right| \leq Ck \lambda^{k-1}, \\
&\quad \left| \frac{\partial^2 \nu_k(x,y)}{\partial x^2} \right| \leq Ck^2 \lambda^{k-2}, \quad \left| \frac{\partial^2 \nu_k(x,y)}{\partial y^2} \right| \leq Ck^2 \lambda^{k-2} \quad \text{and} \quad \left| \frac{\partial^2 \nu_k(x,y)}{\partial x \partial y} \right| \leq Ck^2 \lambda^{k-2}.
\end{align*}
\end{lemmaB}

\begin{proof}
Since there exists a $1-\gamma<\lambda <1$ and finite $C$ such that for all $k$,
$$
\sup _s\left\|A(s)^k\right\|_{\text {spec }} \leq C \lambda^k,
$$
 we have
$$
\left\|Z_t^2(s)\right\|_2 \leq  \sum_{k=0}^{\infty}\left\|A(u)^k\right\|_{\text {spec }} \left\|\zeta_{t-k}(s)\right\|_2 \leq C \sum_{k=0}^{\infty} \lambda^k\left\|\zeta_{h,-i}(u)\right\|_2.
$$

Furthermore, since   $ \sup_s\left|\nu_k(s) \right| \leq K \sup_s \left\|A(s)^k\right\|,$ $(1)$ follows directly.
Further, since $\mathbb{E}\left(\left| \eta_t^{4+v}(s)\right| \right)<\infty$, $(2)$ can easily be verified.
We now prove $(3)$. Using $\nu_k(s)= [A(s)^k]_{(1, 1)}$ and the norm inequalities, we have
$$
 \left|\nu_k(s)-\nu_k(v)\right| \leq \left\|A(s)^k -A(v)^k \right\|_2  
.$$

By the Lipschitz continuity of $ \omega(s)$ , $\alpha_i(s)$ and $\beta_j(s)$ we have that $$\|A(s)-A(v)\|_2 \leq C\|s-v\|_\infty .$$ Hence, by using $\sup _s\left\|A(s)^k\right\|_{\text {spec}} \leq C \lambda^k $ we obtain
$$
\left\|A(s)^k-A(v)^k\right\|_2 \leq \sum_{i=0}^{k-1}\left\|A(s)^{i}\right\|_{\text {spec }} \| A(s)- A(v)\|_2\left\| A(v)^{k-i-1}\right\|_{ \text {spec }} \leq C k \lambda^{k-1}\|s-v\|_{\infty},
$$
for some $C$ independent of $s, v$. This establishes $(3)$.
Finally, we prove $(4)$ we consider the partial derivatives of $\nu_k(s)$. From the definition of $\nu_k(s)$ we have
$$
\left|\frac{\partial \nu_k(s))}{\partial x}\right|  \leq\left\|\frac{\partial A(s)^k }{\partial x}\right\|_2 
$$

 we require a bound for $\left\|\frac{ A(s)^k}{\partial x}\right\|_2$. Expanding $\frac{A(s)^k}{\partial x}$ gives
$$
\left\|\frac{\partial A(s)^k}{\partial x}\right\|_2=\left\|\sum_{i=0}^{k-1} A(s)^i \frac{\partial A(u)}{\partial x} A(s)^{k-i-1}\right\|_2
$$

By using the uniform boundedness of $\frac{ \omega(s)}{\partial x}, \; \frac{ \alpha_i(s)}{\partial x}\;  \text{and}  \; \frac{ \beta_j(s)}{\partial x}$ we have
$$
\left\|\frac{\partial A(s)^k}{\partial x}\right\|_2 \leq C \sum_{i=0}^{k-1}\left(\left\|A(s)^i\right\|_{\text {spec }} \left\|\frac{\partial A(s)}{\partial x}\right\|_2 \left\|A(s)^{k-i-1}\right\|_{\text {spec }}\right) \leq C k \lambda^{k-1}
$$

 we have $\left| \frac{\partial \nu_k(x,y)}{\partial x } \right| \leq C k \lambda^{k-1}$. By applying the same method, we obtain the bounds for all partial derivatives of $\nu_k(s)$ up to the second order.

\end{proof}

The theorem presented below illustrates that the spatially stationary process $ \widetilde{Z}^2_{t}(s)$ offers a local approximation of the spatially nonstationary process $Z^2_t(s)$ within a certain neighborhood of $s_0$.
\begin{theoremB}\label{spa}
Let $Z_t(s)$ be  STGARCH model  and satisfies Assumptions \ref{assumptions}. $\widetilde{Z}^2_{t}(s)$ is defined as follow:$$ \widetilde{Z}_{t}^2(s) =\sum_{k=0}^{\infty} \nu_k(s_0) \zeta_{t-k}(s) + C(s_0),$$

  Then we have
$$ |Z_t^2(s) - \widetilde{Z}^2_{t}(s)| \leq \|s- s_0\|_{\infty} W_t(s),  $$
where for some $1 - \gamma \leq \lambda < 1$,
$$  W_t(s) = C\sum_{k=0}^{\infty}(k + 1)^2 \lambda^k |\zeta_{t-k}(s)| + C_1, $$

$C$ and $ C_1$ are  a finite constant independent of s, and $W_t(s)$ is a stationary process in both space and time with a finite variance.
\end{theoremB}
\begin{proof}
 By using the MA( $\infty$ ) representation, we can  evaluate an upper bound for $\left|Z_t^2(s) - \widetilde{Z}^2_{t}(s)\right|$. From the definition of $\widetilde{Z}^2_{t}(s)$,  Lemma \ref{lem}(3) and Lipschitz continuity,  we have
$$
\left|Z_t^2(s) - \widetilde{Z}^2_{t}(s)\right|  \leq \sum_{k=0}^{\infty} \left|\nu_k(s)-\nu_k\left(s_0\right)\right|\left|\zeta_{t-k}(s)\right| + \left |C(s) - C(s_0)  \right|$$ 
$$ \leq C\left\|s-s_0\right\|_{\infty} \sum_{k=0}^{\infty} (k + 1)^2 \lambda^k\left|\zeta_{t-k}(s)\right|  + C_1 \left\|s-s_0\right\|_{\infty}.$$
Therefore,
$$\left|Z_t^2(s) - \widetilde{Z}^2_{t}(s)\right|\leq\left\|s-s_0\right\|_{\infty} W_t(s),$$
where $ W_t(s) = C\sum_{k=0}^{\infty}(k + 1)^2 \lambda^k |\zeta_{t-k}(s)| + C_1 $. 
\end{proof}
\begin{corollaryB}\label{cor}
Let $Z_t(s)$ be  STGARCH model  and satisfies  Assumptions \ref{assumptions}, and $$Z_t^2(s) = \sigma_t^2(s) \eta_t^2(s). $$
if $$ |Z_t^2(s) - \widetilde{Z}^2_{t}(s)| \leq \|s- s_0\|_{\infty} W_t(s),  $$
then 
$$ |\sigma_t^2(s) - \widetilde{\sigma}^2_{t}(s)| \leq \|s- s_0\|_{\infty} N_t(s),  $$ where 
$$  N_t(s) = \frac{W_t(s)}{\eta_t^2(s)} $$
 $N_t(s)$ is a stationary process in both space and time with a finite variance.
 \end{corollaryB}

\end{document}